\documentclass[graybox]{svmult}

\usepackage{type1cm}        
\usepackage{makeidx}         
\usepackage{graphicx}        
\usepackage{multicol}        
\usepackage[bottom]{footmisc}
\usepackage{appendix}
\usepackage{url}

\usepackage{newtxtext}       %
\usepackage[varvw]{newtxmath}       
\usepackage[all]{xy}
\usepackage{amscd}
\usepackage{amsmath}
\usepackage{bm}

\usepackage[
        hypertexnames=false,
        hyperindex,
        pagebackref,
        bookmarks=false,
        colorlinks=true,
        linkcolor=blue,
        citecolor=blue,
]{hyperref}

\makeindex             

\begin{document}

\title*{On Hilbert-Poincar\'{e} series of affine semi-regular polynomial sequences and related Gr\"{o}bner bases}
\titlerunning{Hilbert-Poincar\'{e} series of affine semi-regular sequences and related Gr\"{o}bner bases}
\author{Momonari Kudo and\\ Kazuhiro Yokoyama}
\institute{Momonari Kudo \at Fukuoka Institute of Technology, 3-30-1 Wajiro-higashi, Higashi-ku, Fukuoka, 811-0295 Japan\\ \email{m-kudo@fit.ac.jp}
\and Kazuhiro Yokoyama \at Rikkyo University, 3-34-1 Nishi-Ikebukuro, Toshima-ku, Tokyo, 171-8501 Japan\\ \email{kazuhiro@rikkyo.ac.jp}}
%
%
\maketitle

\abstract{
Gr\"{o}bner bases are nowadays central tools for solving various problems in commutative algebra and algebraic geometry.
A typical use of Gr\"{o}bner bases is the multivariate polynomial system solving, which enables us to construct algebraic attacks against post-quantum cryptographic protocols.
Therefore, the determination of the complexity of computing Gr\"{o}bner bases is very important both in theory and in practice:
One of the most important cases is the case where input polynomials compose an (overdetermined) affine semi-regular sequence.
The first part of this paper aims to present a survey on Gr\"{o}bner basis computation and its complexity.
In the second part, we shall give an explicit formula on the (truncated) Hilbert-Poincar\'{e} series associated to the homogenization of an affine semi-regular sequence.
Based on the formula, we also study (reduced) Gr\"{o}bner bases of 
the ideals generated by an affine semi-regular sequence and  its {homogenization}.
Some of our results are considered to give mathematically rigorous proofs of the correctness of methods for computing Gr\"{o}bner bases of the ideal generated by an affine semi-regular sequence.
}

\section{Introduction}

Let $K$ be a field, and $R = K [x_1, \ldots , x_n]$ the polynomial ring in $n$ variables over $K$.
For a polynomial $f$ in $R$, let $f^{\rm top}$ denote its maximal total degree part {which is called the {\em top part} of $f$ here}, 
and let $f^h$ denote its homogenization in $R'=R[y]$ by an extra variable $y$, see Subsection \ref{sec:homogenization} below for details.
{We denote by $\langle F \rangle_R$ (or $\langle F \rangle$ simply) the ideal generated by a non-empty subset $F$ of $R$.}
For a finitely generated graded $R$-(or $R'$-)module $M$, we {also} denote by ${\rm HF}_M$ and ${\rm HS}_M$ its {Hilbert function} and its {Hilbert–Poincar\'{e} series}, respectively.
A {\it Gr\"{o}bner basis} of an ideal $I$ in $R$ is defined as a special kind of generating set for $I$, and it gives a computational tool to determine many properties of the ideal $I$.
A typical application of computing Gr\"{o}bner bases is solving the multivariate polynomial (MP) problem:
Given $m$ polynomials $f_1, \ldots , f_m$ in $R$, find $(a_1,\ldots , a_n) \in K^n$ such that $f_i(a_1,\ldots,a_n)=0$ for all $i$ with $1 \leq i \leq m$.
A particular case where polynomials are all quadratic is called {the MQ problem}, and its hardness is applied to constructing public-key cryptosystems and digital signature schemes that are expected to be quantum resistant.
Therefore, analyzing the complexity of computing Gr\"{o}bner bases is one of the most important problems both in theory and in practice.

An algorithm for computing Gr\"{o}bner bases was proposed first by Buchberger~\cite{Buchberger}, and so far a number of its improvements such as the $F_4$~\cite{F4} and $F_5$~\cite{F5} algorithms have been proposed, see Subsection \ref{sec:GB} below for a summary.
In general, it is very difficult to determine the complexity of computing Gr\"{o}bner bases, but in some cases, we can estimate it with several algebraic invariants such as the solving degree, the degree of regularity, the Castelnuovo–Mumford regularity, and the first and last fall degrees; we refer to \cite{CG23} for the relations between these invariants.

The first part of this paper aims to survey Gr\"{o}bner basis computation, and to review its complexity in the case where input polynomials generate a zero-dimensional ideal.
For this, in Section \ref{sec:pre}, we first recall foundations in commutative algebra such as Koszul complex, Hilbert-Poincar\'{e} series, and semi-regular sequence, which are useful ingredients to estimate the complexity of computing Gr\"{o}bner bases.
Then, we overview existing Gr\"{o}bner basis algorithms in Subsection \ref{sec:GB}.
Subsequently, it will be described in Subsection \ref{subsec:complexity} how to estimate the complexity of computing the reduced Gr\"{o}bner basis of a zero-dimensional ideal, with the notion of {homogenization}.

In the second part, we focus on {\it affine semi-regular} polynomial sequences, where a sequence $\bm{F} = (f_1, \ldots, f_m) \in R^m$ of (not necessarily homogeneous) polynomials is said to be affine (cryptographic) semi-regular if $\bm{F}^{\rm top} = (f_1^{\rm top}, \ldots , f_m^{\rm top})$ is (cryptographic) semi-regular, see Definitions \ref{def:semiregB}, \ref{def:csemireg}, and \ref{def:affine_semireg} for details.
Note that homogeneous semi-regular sequences are conjectured by Pardue~\cite[Conjecture B]{Pardue} to be generic sequences of polynomials, and affine (cryptographic) semi-regular sequences are often appearing in the construction of multivariate public key cryptosystems and digital signature schemes.
In Section \ref{sec:main} below, we relate the Hilbert-Poincar\'{e} series of $R' / \langle \bm{F}^h \rangle$ with that of $R/\langle \bm{F}^{\rm top} \rangle$.
As a corollary, we obtain an explicit formula of the truncation at degree $D-1$ of the Hilbert-Poincar\'{e} series of $R' / \langle \bm{F}^h \rangle$, where $D$ is the degree of regularity for $\langle \bm{F}^{\rm top} \rangle$.
{The following theorem summarizes these results:}

\begin{theorem}[Theorem \ref{thm:main}, Corollaries \ref{cor:Dreg} and \ref{cor:LM}]
With notation as above, assume that $\bm{F}$ is affine cryptographic semi-regular.
Then ${\rm HF}_{R'/\langle \bm{F}^h \rangle}(d) = \sum_{i=0}^d{\rm HF}_{R/\langle \bm{F}^{\rm top} \rangle}(i)$ and $(\langle \mathrm{LM}(\langle \bm{F}^h \rangle) \rangle_{R'})_d = (\langle \mathrm{LM}(\langle \bm{F}^{\rm top} \rangle) \rangle_{R'})_d $ for each $d$ with $d < D$, \textcolor{black}{where we use a DRL ordering on the set of monomials in $R$ and its homogenization on that in $R'$.}
Hence, we also obtain ${\rm HS}_{R'/\langle \bm{F}^h\rangle}(z)\equiv \prod_{i=1}^m (1-z^{d_i})/(1-z)^{n+1}\pmod{z^D}$, so that $\bm{F}^h$ is $D$-regular (see Definition \ref{def:semiregB} for the definition of $d$-regularity).
\end{theorem}

As an application of this theorem, we explore reduced Gr\"{o}bner bases of $\langle \bm{F}\rangle$, $\langle \bm{F}^h\rangle$, and $\langle \bm{F}^{\rm top}\rangle$ in Section \ref{sec:app} below, dividing the cases into the degree less than $D$ or not.
In particular, we rigorously prove some existing results, which are often used for analyzing the complexity of computing Gr\"{o}bner bases, and moreover extend them \textcolor{black}{to our case}.

\section{Preliminaries}\label{sec:pre}
In this section, we recall definitions of Koszul complex, Hilbert–Poincar\'{e} series, and semi-regular polynomial sequences, and collect some known facts related to them.
Throughout this section, let $R = K[X] = K[x_1, \ldots , x_n]$ be the polynomial ring of $n$ variables {$X=(x_1,\ldots , x_n)$} over a field $K$.
As a notion, for a polynomial $f$ in $R$, we denote its total degree by $\deg(f)$. 
As $R$ is a graded ring with respect to total degree, 
for a polynomial $f$, its maximal total degree part, denoted by $f^{\rm top}$, is 
defined as its graded component of $\deg(f)$, that is, 
the sum of all terms of $f$ whose total degree equals to $\deg(f)$. 

\subsection{Koszul complex and its homology}

Let $f_1, \ldots , f_m \in R$ be homogeneous polynomials of degrees $d_1, \ldots , d_m$, and put $d_{j_1 \cdots j_i} := \sum_{k=1}^i d_{j_k}$.
For each $0 \leq i \leq m$, we define a free $R$-module {of rank $\binom{m}{i}$}
\[
K_i (f_1,\ldots , f_m) := 
\left\{
\begin{array}{cl}
\displaystyle \bigoplus_{1 \leq j_1 < \cdots < j_i \leq m} R({-d_{j_1\cdots j_i}}) \mathbf{e}_{j_1 \cdots j_i} & (i \geq 1)\\
R & (i=0), 
\end{array}
\right.
\]
where $\mathbf{e}_{j_1 \cdots j_i}$ is a standard basis.
We also define a graded homomorphism
\begin{equation*}
    \varphi_i : K_i (f_1,\ldots , f_m) \longrightarrow K_{i-1} (f_1,\ldots , f_m)
\end{equation*}
of degree $0$ by
\begin{equation*}
    \varphi_i (\mathbf{e}_{j_1\cdots j_i}) := \sum_{k=1}^i (-1)^{k-1} f_{j_k} \mathbf{e}_{j_1 \cdots \hat{j_k} \cdots j_i}.
\end{equation*}
Here, $\hat{j_k}$ means to omit $j_k$.
For example, we have $\mathbf{e}_{1\hat{2}3} = \mathbf{e}_{13}$.
To simplify the notation, we set $K_i := K_i(f_1,\ldots , f_m)$.
Then,
\begin{equation}\label{eq:Koszul}
K_{\bullet} : 0 \to K_{m} \xrightarrow{\varphi_m}\cdots \xrightarrow{\varphi_3} K_2 \xrightarrow{\varphi_2} K_1 \xrightarrow{\varphi_1} K_0 \to 0
\end{equation}
is a complex, and we call it the {\it Koszul complex} \textcolor{black}{on} $(f_1,\ldots,f_m)$.
The $i$-th homology group \textcolor{black}{of $K_{\bullet}$} is given by
\begin{equation*}
    H_i (K_{\bullet}) = \mathrm{Ker}(\varphi_i) / \mathrm{Im}(\varphi_{i+1}).
\end{equation*}
In particular, we have
\begin{equation*}
    H_0 (K_{\bullet}) = R/\langle f_1, \ldots , f_m \rangle_R.
\end{equation*}
The kernel and the image of a graded homomorphism are both graded submodules in general, so that $\mathrm{Ker}(\varphi_i)$ and $\mathrm{Im}(\varphi_{i+1})$ are graded $R$-modules, and so is the quotient module $H_i (K_{\bullet})$.
In the following, we denote by $H_i (K_{\bullet})_d$ the degree-$d$ homogeneous part of $H_i (K_{\bullet})$.

Note that $\mathrm{Ker}(\varphi_1) = \mathrm{syz}(f_1,\ldots , f_m)$ (the right hand side is the module of syzygies), and that $\mathrm{Im}(\varphi_2) \subset K_1 = \bigoplus_{j=1}^m R(-d_j)\mathbf{e}_j$ is generated by
\begin{equation*}
    \{ \mathbf{t}_{i,j} := f_i \mathbf{e}_j - f_j \mathbf{e}_{i} : 1 \leq i < j \leq m \}.
\end{equation*}
Hence, putting
\begin{equation*}
    \mathrm{tsyz}(f_1,\ldots ,f_m) := \langle \mathbf{t}_{i,j} : 1 \leq i < j \leq m \rangle_R,
\end{equation*}
we have
\begin{equation}\label{eq:H1}
    H_1(K_{\bullet}) = \mathrm{syz}(f_1,\ldots , f_m)/ \mathrm{tsyz}(f_1,\ldots , f_m).
\end{equation}

\begin{definition}[Trivial syzygies]\label{def:trivial}
With notation as above, we call each generator $\mathbf{t}_{i,j}$ (or each element of $\mathrm{tsyz}(f_1,\ldots , f_m)$) a {\it trivial syzygy} for $(f_1, \ldots , f_m)$.
We also call $ \mathrm{tsyz}(f_1,\ldots , f_m)$ the {\it module of trivial syzygies}. 
\end{definition}

We also note that $H_m(K_{\bullet}) = 0$, since $\varphi_m$ is clearly injective by definition.

\begin{remark}\label{rem:ff}
    When $K = \mathbb{F}_q$, a vector of the form $f_{i}^{q-1}\mathbf{e}_i$ is also referred to as a trivial syzygy, in the context of Ding-Schmidt's definition for {\it first fall degree}~\cite{Ding2013} (see \cite[Section 4.2]{CG20} or \cite[Section 3.2]{N22} for reviews).
    More concretely, putting $B := R/\langle x_1^q, \ldots , x_n^q \rangle_R$ and $\overline{f}_i := f_i \bmod{\langle x_1^q, \ldots , x_n^q \rangle}$, we define the Koszul complex on $(\overline{f}_1,\ldots , \overline{f}_m) \in B^m$ similarly to that on $(f_1, \ldots , f_m) \in R^m$, and denote it by $\overline{K}_{\bullet} = \overline{K}_{\bullet}(\overline{f}_1,\ldots , \overline{f}_m )$.
    Then, the vectors $\overline{f}_i \mathbf{e}_j - \overline{f}_j \mathbf{e}_i$ and $\overline{f}_i^{q-1} \mathbf{e}_i$ in $B^m$ for $1 \leq i < j \leq m$ are syzygies for $(\overline{f}_1,\ldots , \overline{f}_m)$.
    Each $\overline{f}_i \mathbf{e}_j - \overline{f}_j \mathbf{e}_i$ is called a Koszul syzygy, and the Koszul syzygies together with $\overline{f}_i^{q-1} \mathbf{e}_i$'s are referred to as trivial syzygies for $(\overline{f}_1, \ldots , \overline{f}_m)$.
    The {\it first fall degree} $d_{\rm ff}(f_1,\ldots, f_m)$ is \textcolor{black}{defined as} the minimal integer $d$ with $\mathrm{syz}(\overline{f}_1,\ldots , \overline{f}_m)_d \supsetneq \mathrm{tsyz^+}(\overline{f}_1,\ldots , \overline{f}_m)_d$ in $(B_{d-d_i})^m$, where $\mathrm{tsyz^+}(\overline{f}_1,\ldots , \overline{f}_m)$ denotes the submodule in $B^m$ generated by the trivial syzygies for $(\overline{f}_1, \ldots , \overline{f}_m)$.
    
    Note that, for each $i$, a homomorphism $H_i(K_{\bullet}) \to H_i (\overline{K}_{\bullet})$ is canonically induced by taking modulo $\langle x_1^q,\ldots , x_n^q\rangle_R$.
    In particular, we have the following composite $K$-linear map:
    \[
    \eta_d : H_1(K_{\bullet})_d \to H_1 (\overline{K}_{\bullet})_d \to \mathrm{syz}(\overline{f}_1,\ldots , \overline{f}_m)_d / \mathrm{tsyz^+}(\overline{f}_1,\ldots , \overline{f}_m)_d.
    \]
    for each $d$.
    Putting $d = d_{\rm ff}(f_1,\ldots , f_m)$ and letting $D$ to be the minimal integer with $H_1(K_{\bullet})_D \neq 0$, it is straightforward to verify the following:
    \begin{itemize}
        \item If $q > D$, then $\eta_D$ is injective, and $\mathrm{syz}(\overline{f}_1,\ldots , \overline{f}_m)_D \supsetneq \mathrm{tsyz^+}(\overline{f}_1,\ldots , \overline{f}_m)_D$, whence $D \geq d$.
        \item If $q > d$, then $\eta_d$ is surjective.
        In this case, $H_1(K)_d \neq 0$, so that $D \leq d$. 
    \end{itemize}
    See \cite[Lemmas 4.2 and 4.3]{N22} for a proof.
    Therefore, we have $d = D$ for sufficiently large any $q$.
\end{remark}

\subsection{Hilbert–Poincar\'{e} series and semi-regular sequences}\label{subsec:HilSemi}

\begin{definition}[Hilbert–Poincar\'{e} series]
For a finitely generated graded $R$-module $M$, we define the {\it Hilbert function} ${\rm HF}_M$ of $M$, given by
\begin{equation*}
    {\rm HF}_{M}(d) = \mathrm{dim}_K M_d
\end{equation*}
for each $d \in \mathbb{Z}_{\geq 0}$. 
The {\it Hilbert–Poincar\'{e} series} ${\rm HS}_M$ of $M$ is defined as the formal power series
\begin{equation*}
    {\rm HS}_{M}(z) = \sum_{d=0}^{\infty} {\rm HF}_{M}(d) z^d \in \mathbb{Z} \llbracket z \rrbracket.
\end{equation*}
\end{definition}

\begin{theorem}[cf.\ {\cite[Chapter 10]{BW}}]\label{thm:BW}
Let $I$ be a homogeneous ideal of $R$ generated by a set $G \subset R$ of homogeneous elements of degree not greater than a non-negative integer $d$.
\textcolor{black}{Let $\mathrm{LM}(f)$ denote the leading monomial of $f \in R\smallsetminus\{0\}$ with respect to a graded ordering $\prec$ on the set of monomials in $R$.
For a non-empty subset $F \subset R\smallsetminus \{ 0 \}$, put $\mathrm{LM}(F) := \{ \mathrm{LM}(f) : f \in F \}$.}
\textcolor{black}{Then, the following are equivalent:}
\begin{enumerate}
    \item $\langle \mathrm{LM}(G) \rangle_{\leq d} = \langle \mathrm{LM}(I) \rangle_{\leq d}$.
    \item Every $f \in I_{\leq d}$ is reduced to zero modulo $G$. 
    \item For every pair of $f,g \in G$ with $\mathrm{deg}(\mathrm{LCM}(\mathrm{LM}(f),\mathrm{LM}(g))) \leq d$, {the $S$-polynomial $S(f,g)$ is reduced to zero modulo $G$}.
\end{enumerate}
In this case, $G$ is called a $d$-Gr\"{o}bner basis of $I$ with respect to {$\prec$}.
\end{theorem}

We also review the notion of semi-regular sequence defined by Pardue~\cite{Pardue}.

\begin{definition}[Semi-regular sequences, {\cite[Definition 1]{Pardue}}]\label{def:semireg}
Let $I$ be a homogeneous ideal of $R$.
A degree-$d$ homogeneous element $f \in R$ is said to be {\it semi-regular} on $I$ if the multiplication map $(R / I)_{t-d} \longrightarrow (R / I)_{d} \ ; \ g \mapsto g f$ is injective or surjective, for every $t$ with $t \geq d$.
A sequence $(f_1, \ldots , f_m) \in R^m$ of homogeneous polynomials is said to be {\it semi-regular} on $I$ if $f_i$ is semi-regular on $I + \langle f_1, \ldots , f_{i-1} \rangle_R$, for every $i$ with $1 \leq i \leq m$.
\end{definition}

Throughout the rest of this subsection, let $f_1, \ldots , f_m \in R$ be homogeneous elements of degree $d_1, \ldots , d_m$, respectively, and put $I = \langle f_1, \ldots , f_m \rangle_R$.
Furthermore, put $I^{(0)}:= \{ 0 \}$ and $A^{(0)} := R / I^{(0)} = R$.
For each $i$ with $1 \leq i \leq m$, we also set $I^{(i)}:= \langle f_1, \ldots , f_{i} \rangle_R$ and $A^{(i)} := R / I^{(i)}$.
The degree-$d$ homogeneous part $A_d^{(i)}$ of each $A^{(i)}$ is given by $A_d^{(i)} = R_d / I_d^{(i)}$, where $I_d^{(i)} = I^{(i)} \cap R_d$.
We denote by $\psi_{f_i}$ the multiplication map
\begin{equation*}
 A^{(i-1)} \longrightarrow A^{(i-1)} \ ; \ g \mapsto g f_i,
\end{equation*}
which is a graded homomorphism of degree $d_i$.
For every $t \geq d_i$, the restriction map
\[
\psi_{f_i}|_{A^{(i-1)}_{t-d_i}} : A^{(i-1)}_{t-d_i} \longrightarrow A^{(i-1)}_{t}
\]
is a $K$-linear map.
On the other hand, as for the surjective homomorphism
\begin{equation*}
\phi_{i-1} : A^{(i-1)} \longrightarrow A^{(i)} \ ; \ f + I^{(i-1)} \mapsto f + I^{(i)},
\end{equation*}
it is straightforward to see that for each $t$ with $0 \leq t \leq d_i-1$, the restriction map
\[
\phi_{i-1} |_{A_t^{(i-1)}} : A_t^{(i-1)} \longrightarrow A_t^{(i)}
\]
is an isomorphism of $K$-linear spaces, whence
\begin{equation*}
    \mathrm{dim}_K A^{(i-1)}_{t} = \mathrm{dim}_K A^{(i)}_t \quad (0 \leq t \leq d_i -1).
\end{equation*}

\begin{lemma}\label{lem:dimeq}
With notation as above, for each $1 \leq i \leq m$ and for each $t \geq d_i$, we have the following equalities:
\begin{equation}\label{eq:dim11}
    \dim_K A^{(i)}_{t} = \dim_K A^{(i-1)}_{t} - \mathrm{dim}_K \mathrm{Im}\left(A^{(i-1)}_{t-d_i} \xrightarrow{\times f_i} A^{(i-1)}_{t}\right),
\end{equation}
\begin{equation}\label{eq:dim12}
    \mathrm{dim}_K \mathrm{Im}\left(A^{(i-1)}_{t-d_i} \xrightarrow{\times f_i} A^{(i-1)}_{t}\right) = \dim_K A^{(i-1)}_{t-d_i} - \mathrm{dim}_K (0 : f_{i})_{t-d_i},
\end{equation}
where we set $(0: f_i) = \{ g \in A^{(i-1)} : g f_i = 0 \}$.
Hence, 
\begin{itemize}
    \item The multiplication map $A^{(i-1)}_{t-d_i} \xrightarrow{\times f_i} A^{(i-1)}_{t}$ is injective if and only if
    \begin{equation}\label{eq:dim1}
        \dim_K A^{(i)}_{t} = \dim_K A^{(i-1)}_{t} - \dim_K A^{(i-1)}_{t-d_i}.
    \end{equation}
    In this case, one has ${\rm dim}_K A^{(i-1)}_{t-d_i} \leq {\rm dim}_K A^{(i-1)}_{t}$.
    \item The multiplication map $A^{(i-1)}_{t-d_i} \xrightarrow{\times f_i} A^{(i-1)}_{t}$ is surjective if and only if
    \begin{equation}\label{eq:dim2}
    \dim_K A^{(i)}_{t} = 0.
    \end{equation}
    In this case, one has ${\rm dim}_K A^{(i-1)}_{t-d_i} \geq {\rm dim}_K A^{(i-1)}_{t} $.
\end{itemize}
\end{lemma}

\begin{proof}
Let $i$ and $t$ be integers such that $1 \leq i \leq m$ and $t \geq d_i$.
Since we have $(0: f_i)_{t-d_i} = \{ g \in A^{(i-1)}_{t-d_i} : g f_i = 0 \}$, the sequence
\begin{equation*}
0 \longrightarrow (0 : f_i)_{t-d_i} \longrightarrow A^{(i-1)}_{t-d_i} \overset{~~{\times f_{i}}~~}{\longrightarrow}A^{(i-1)}_{t} \longrightarrow  A^{(i)}_{t} \longrightarrow 0
\end{equation*}
of $K$-linear maps is exact, where $(0 : f_i)_{t-d_i} \to A_{t-d_i}^{(i-1)}$ is an inclusion map.
The exactness of this sequence implies the desired equalities \eqref{eq:dim11} and \eqref{eq:dim12}.
\end{proof}

The semi-regularity is characterized by equivalent conditions in Proposition \ref{prop:semireg} below.
In particular, the fourth condition enables us to compute the Hilbert–Poincar\'{e} series of each $A^{(i)}$ easily.

\begin{proposition}[cf.\ {\cite[Proposition 1]{Pardue}}]\label{prop:semireg}
With notation as above, the following are equivalent:
\begin{enumerate}
    \item The sequence $( f_1, \ldots , f_m )\in R^m$ is semi-regular.
    \item For each $1 \leq i \leq m$ and for each $t \geq d_i$, we have \eqref{eq:dim1} or \eqref{eq:dim2}, namely
    \begin{equation*}
    \dim_K A^{(i)}_{t} = \max \{0, \dim_K(A^{(i-1)}_{t}) - \dim_K(A^{(i-1)}_{t-d_i}) \}.
    \end{equation*}
    \item For each $i$ with $1 \leq i \leq m$, we have
    \begin{equation*}
    {\rm HS}_{A^{(i)}}(z) = [{\rm HS}_{A^{(i-1)}}(z) (1-z^{d_i})],
    \end{equation*}
    \textcolor{black}{where $[\cdot]$ means truncating a formal power series over $\mathbb{Z}$ after the last consecutive positive coefficient.}
    \item For each $i$ with $1 \leq i \leq m$, we have
    \begin{equation*}
    {\rm HS}_{A^{(i)}}(z) = \left[ \frac{\prod_{j=1}^{i}(1-z^{d_j})}{(1- z)^n} \right].
    \end{equation*}
\end{enumerate}
\end{proposition}

When $K$ is an infinite field, Pardue also conjectured in \cite[Conjecture B]{Pardue} that generic polynomial sequences are semi-regular.

\subsection{Cryptographic semi-regular sequences}

We here review the notion of {\it cryptographic semi-regular} sequence, which is defined by a 
{condition weaker} than one for semi-regular sequences.
The notion of cryptographic semi-regular sequence is introduced first by Bardet et al.\ (e.g., \cite{BFS}, \cite{BFSY}) motivated to analyze the complexity of computing Gr\"{o}bner bases.
Diem~\cite{Diem2} also formulated cryptographic semi-regular sequences, in terms of commutative and homological algebra.
The terminology ``cryptographic'' was named by Bigdeli et al.\ in their recent work~\cite{BNDGMT21}, in order to distinguish such a sequence from a semi-regular one defined by Pardue (see Definition \ref{def:semireg} in the previous subsection).

\begin{definition}[{\cite[Definition 3]{BFS}}; see also {\cite[Definition 1]{Diem2}}]\label{def:semiregB}
Let $f_1, \ldots , f_m \in R$ be homogeneous polynomials of positive degrees $d_1, \ldots , d_m$ respectively, and put $I = \langle f_1, \ldots , f_m \rangle_R$.
The notations $I^{(i)}$ and $A^{(i)}$ are also the same as in the previous subsection.
For each integer $d$ with $d \geq \mathrm{max}\{d_i : 1 \leq i \leq m \}$, we call a sequence $(f_1, \ldots , f_m) \in R^m$ of homogeneous polynomials {\it $d$-regular} if it satisfies the following condition:
\begin{itemize}
    \item For each $i$ with $1 \leq i \leq m$, if a homogeneous polynomial $g \in R$ satisfies $g f_i  \in \langle f_1, \ldots , f_{i-1} \rangle_R$ and $\mathrm{deg}(g f_i) < d$, then we have $g \in \langle f_1, \ldots , f_{i-1} \rangle_R$.
    In other word, the multiplication map $A^{(i-1)}_{t-d_i} \longrightarrow A^{(i-1)}_{t} \ ; \ g \mapsto g f_i$ is injective for every $t$ with $d_i \leq t < d$.
\end{itemize}
\end{definition}

Diem~\cite{Diem2} determined the (truncated) Hilbert-Poincar\'{e} series of $d$-regular sequences as in the following proposition:

\begin{theorem}[cf.\ {\cite[Theorem 1]{Diem2}}]\label{lem:Diem2}
With the same notation as in Definition \ref{def:semiregB}, the following are equivalent 
{\color{black} for each $d$ with $d \geq \mathrm{max}\{d_i : 1 \leq i \leq m \}$:}
\begin{enumerate}
    \item The sequence $( f_1, \ldots , f_m )\in R^m$ is $d$-regular.
    Namely, for {each} $(i,t)$ with $1 \leq i \leq m$ and $d_i \leq t < d$, the equality \eqref{eq:dim1} holds.
    \item We have
    \begin{equation*}
    {\rm HS}_{A^{(m)}}(z) \equiv \frac{\prod_{j=1}^{m}(1-z^{d_j})}{(1- z)^n} \pmod{z^d}.
    \end{equation*}
    \item $H_1 (K_{\bullet}(f_1, \ldots , f_m))_{\leq d-1} = 0$.
\end{enumerate}
\end{theorem}

\begin{proposition}[{\cite[Proposition 2 (a)]{Diem2}}]\label{prop:Diem2}
With the same notation as in Definition \ref{def:semiregB}, let $D$ and $i$ be natural numbers.
Assume that {$H_i (K(f_1,\ldots , f_m))_{\leq D} = 0$}.
Then, for each $j$ with $1 \leq j < m$, we have 
{$H_i(K(f_1, \ldots , f_j))_{\leq D} = 0$}.
\end{proposition}

\begin{definition}
A finitely generated graded $R$-module $M$ is said to be {\it Artinian} if there exists a sufficiently large $D \in \mathbb{Z}$ such that $M_d = 0$ for all $d \geq D$.
\end{definition}

\begin{definition}[{\cite[Definition 4]{BFS}}, {\cite[Definition 5]{BFSY}}]\label{def:dreg}
For a homogeneous ideal $I$ of $R$, we define its {\it degree of regularity} $d_{\rm reg}(I)$ as follows:
If the finitely generated graded $R$-module $R/I$ is Artinian, we set $d_{\rm reg} (I) := \mathrm{min} \{ d : R_d = I_d \}$, and otherwise we set $d_{\rm reg}(I) := \infty$.
\end{definition}

{As for an upper-bound on the degree of regularity, we refer to \cite[Theorem 21]{GG23-2}.}

\begin{remark}
    In Definition \ref{def:dreg}, since $R/I$ is Noetherian, it is Artinian if and only if it is of finite length.
    In this case, the degree of regularity $d_{\rm reg}(I)$ is equal to the {\it Castelnuovo-Mumford regularity} $\mathrm{reg}(I)$ of $I$ {(see e.g., \cite[\S20.5]{Eisen} for the definition)}, whence $d_{\rm reg}(I) = \mathrm{reg}(I) = \mathrm{reg}(R/I) + 1$.
\end{remark}

\begin{definition}[{\cite[Definition 5]{BFS}}, {\cite[Definition 5]{BFSY}}; see also {\cite[Section 2]{Diem}}]\label{def:csemireg}
A sequence $(f_1, \ldots , f_m) \in R^m$ of homogeneous polynomials is said to be {\it cryptographic semi-regular} if it is $d_{\rm reg}(I)$-regular, where {we set} $I = \langle f_1, \ldots , f_m \rangle_R$.
\end{definition}

The cryptographic semi-regularity is characterized by equivalent conditions in Proposition \ref{prop:Diem} below.
In particular, the second condition enables us to compute the Hilbert–Poincar\'{e} series of $A^{(i)}$ easily.

\begin{proposition}[{\cite[Proposition 1 (d)]{Diem2}}; see also {\cite[Proposition 6]{BFSY}}]\label{prop:Diem}
With the same notation as in Definition \ref{def:semiregB}, we put $D=d_{\rm reg}(I)$.
Then, the following are equivalent:
\begin{enumerate}
    \item $( f_1, \ldots , f_m )\in R^m$ is cryptographic semi-regular.
    \item We have
    \begin{equation}\label{eq:semiregHil2}
    {\rm HS}_{R/I}(z) = \left[ \frac{\prod_{j=1}^{m}(1-z^{d_j})}{(1- z)^n} \right].
    \end{equation}
    \item $H_1 (K_{\bullet}(f_1, \ldots , f_m))_{\leq D-1} = 0$.
\end{enumerate}
\end{proposition}

{
\begin{remark}\label{rm:dreg}
By the definition of {\em degree of regularity}, if $(f_1,\ldots,f_m)\in R^m$ is cryptographic semi-regular, 
$d_{\rm reg}(I)$ coincides with $\deg({\rm HS}_{R/I}(z))+1$, where we set $I=\langle f_1,\ldots,f_m\rangle$. 
\end{remark}
}

{In 1985, Fr\"{o}berg already conjectured in \cite{Froberg} that, when $K$ is an infinite field, a generic sequence of homogeneous polynomials $f_1,\ldots,f_m \in R$ of degrees $d_1,\ldots , d_m$ generates an ideal $I$ with the Hilbert-Poincar\'{e} series of the form \eqref{eq:semiregHil2}, namely $(f_1,\ldots , f_m)$ is cryptographic semi-regular.
It can be proved (cf.\ \cite{Pardue}) that Fr\"{o}berg's conjecture is equivalent to Pardue's one~\cite[Conjecture B]{Pardue} introduced in Subsection \ref{subsec:HilSemi}.
We also note that Moreno-Soc\'{i}as conjecture~\cite{MS} is stronger than the above two conjectures, see \cite[Theorem 2]{Pardue} for a proof.}

It follows from the fourth condition of Proposition \ref{prop:semireg} together with the second condition of Proposition \ref{prop:Diem} that the semi-regularity implies the cryptographic semi-regularity.

\begin{definition}[Affine semi-regular sequences]\label{def:affine_semireg}
    A sequence $\bm{F}= (f_1,\ldots , f_m)\in R^m$ of not necessarily homogeneous polynomials $f_1,\ldots , f_m$ is said to be semi-regular (resp.\ cryptographic semi-regular) if $\bm{F}^{\rm top} = (f_1^{\rm top},\ldots, f_m^{\rm top})$ is semi-regular (resp.\ cryptographic semi-regular).
    In this case, we call $\bm{F}$ an {\it affine semi-regular (resp.\ affine cryptographic semi-regular)} {sequence}.
\end{definition}

\begin{remark}
     For an affine \textcolor{black}{cryptographic} semi-regular sequence $\bm{F}= (f_1,\ldots , f_m)\in R^m$ with $K=\mathbb{F}_q$, it follows from Proposition \ref{prop:Diem} that $d_{\rm reg} (\langle \bm{F}^{\rm top} \rangle) \leq d_{\rm ff}(f_1^{\rm top},\ldots , f_m^{\rm top})$ for $q \gg 0$, where $d_{\rm ff} (f_1^{\rm top},\ldots , f_m^{\rm top})$ is the first fall degree defined in Remark \ref{rem:ff}.
\end{remark}

\section{Quick review on the computation of Gr\"{o}bner basis}

In this section, we first review previous studies on the computation of Gr\"{o}bner bases for polynomial ideals.

\subsection{Overview of existing Gr\"{o}bner basis algorithms}\label{sec:GB}

Since Buchberger \cite{Buchberger} discovered the notion of Gr\"obner basis and 
{a} fundamental algorithm for computing them, 
many efforts have been done for improving the efficiency 
of Gr\"{o}bner basis computation based on Buchberger's algorithm. 
In his algorithm, S-polynomials {play} an important role 
for Gr\"obner basis computation and give a famous termination 
criterion called Buchberger's criterion, 
that is, for a given ideal $I$ {of a polynomial ring over a field}, its finite generating subset $G$ is a Gr\"obner basis of $I$ with respect to a monomial ordering if and only if the S-polynomial 
$S(g,g')$ for any distinct pair $g,g'\in G$ is reduced to $0$ {modulo} $G$. 
For details on Buchberger's algorithm and monomial orderings, 
see e.g., \cite{BW}. 

In the below, we list effective improvements for algorithms 
which are, at the same time, very useful to analyze the 
complexity of Gr\"obner basis computation. 
Here we note that the choice of a monomial ordering 
is also very crucial for the efficiency of Gr\"ober basis 
computation, but we here do not discuss about its choice. 
(In general, the degree reverse lexicographical (DRL) ordering\footnote{This ordering is also called the graded reverse lexicographical (grevlex) ordering.} 
is considered as the most efficient ordering for the computation.)
\begin{description}
\item{(1) } {\bf Related to S-polynomial:} 
\begin{description}
\item{(1-1) }{\bf Strategy for selecting S-polynomial:} 
It is considered to be very effective to apply the {\it normal strategy}, 
where we choose a pair $(g,g')$ for which the least common multiple (LCM) of the 
leading monomials ${\rm LM}(g)$ and ${\rm LM}(g')$ with respect to the 
fixed ordering $\prec$ as smaller as possible. 
(See \cite[Chapter 5.5]{BW}.) 
The strategy is very suited for a homogeneous ideal 
with a {\em graded}\footnote{\textcolor{black}{We also call a graded ordering a {\it degree-compatible} ordering.}} ordering such as DRL, 
as we can {utilize} the graded structure of a homogeneous ideal. 
Also, the {\em sugar strategy} is designed for a non-homogeneous 
ideal generated by $F$ to make the computational behavior very close to 
that for the ideal generated by the {\it homogenization} $F^h$.
See {Subsection} \ref{sec:homogenization} {below} for some details on homogenization. (See also \cite[Chapter 2.10]{CLO}.)
\item{(1-2) }{\bf Avoiding unnecessary S-polynomial computation:} 
In Buchberger's algorithm, we add a polynomial to a generating set $G$
which is computed from an S-polynomial by possible reduction of 
elements in $G$. Since the cost of the construction of S-polynomials and their reduction dominate the whole computation, 
S-polynomials which are reduced to 0 are very harmful for the 
efficiency. Thus, it is highly desired to 
avoid such unnecessary S-polynomials as many as possible. 

\begin{description}
\item{(A)  }{\bf Based on simple rules:} 
At earlier stages, there are easily computable rules, 
called Buchberger's criterion and Gebauer-M\"oller's one.
Those are using the relation of the LMs of a pair 
and those of a triple, see \cite[Chapter 5.5]{BW}. 
Then, in 2002, Faug\'ere \cite{F5} introduced the notion of {\em signature} 
and proposed his $F_5$ algorithm based on
a general rule among signatures.  We call algorithms using 
signatures including variants of $F_5$ {\em signature-based {algorithms} (SBA)}.
See a survey \cite{EF} and Subsection \ref{sec:sba} below for details.
\item{(B) }{\bf Using invariants of polynomial ideal:} 
For a homogeneous ideal $I$ of a polynomial ring $R$, 
when its Hilbert function ${\rm HF}_{R/I}(z)$ is known 
before the computation, we can utilize the value ${\rm HF}_{R/I}(d)$ for each 
$d\in {\mathbb N}$ {(cf.\ \cite{Tra})}. 
Because, by the value ${\rm HF}_{R/I}(d)$, we can check whether 
we can stop the computation of S-polynomials of degree $d$ or not. 
We call an algorithm using Hilbert functions a {\em Hilbert driven} (Buchberger's)
algorithm. See \cite{Tra}, \cite[Chapter 10.2]{CLO} or \cite[Section 3.5]{DL}. 
\end{description}
\end{description}
\item
{(2) }{\bf Efficient computation of S-polynomial reduction:} 
Since the computation of S-polynomial reduction is a dominant step 
in the whole Gr\"obner basis computation, 
its efficiency heavily affects the total efficiency. 
As the reduction for a polynomial by elements of $G$ 
is sequentially applied, we can transform the whole reduction 
to a Gaussian elimination of a matrix. 
This approach was suggested in form of Macaulay matrices by 
Lazard \cite{Lazard81} and the first efficient algorithm was 
given by Faug\'ere \cite{F4}, which is called the $F_4$ algorithm. 
Of course, 
we can combine the $F_4$ and $F_5$ algorithms effectively, which is called the 
{\em matrix-$F_5$ algorithm}. 
\item
{(3) } {\bf Solving coefficient growth:} 
For a polynomial ideal over the rational number field ${\mathbb Q}$, 
the computation may be suffered by certain growth of coefficients 
in polynomials appearing during Gr\"obner basis computation. 
To resolve this problem, several modular methods were proposed. 
As a typical one, we can use Chinese remainder algorithm (CRA), where 
we first compute the reduced Gr\"obner bases $G_p$ over several 
finite fields ${\mathbb F}_p$ and then recover the reduced Gr\"obner basis 
from $G_p$'s by CRA. See \cite{NY} for details about choosing primes $p$. 
\end{description}

\begin{remark}
For several public {key} cryptosystems based on 
polynomial ideals over finite fields or {the} elliptic curve discrete logarithm {problem}, 
estimating the cost of finding zeros of polynomial ideals is important to analyze the security of those systems, where the computation of their Gr\"obner bases is a fundamental tool. 
In this situation, {the} $F_5$ algorithm and {matrix-}$F_5$ algorithm as its 
efficient variant with an efficient DRL ordering are considered, 
as not only those can attain efficient computation but also 
they are suited for estimating the computational complexity. 
\end{remark}

In the following, we introduce the notion of {\em homogenization} and 
an algorithm for Gr\"obner basis computation based on 
{\rm signatures} ($F_5$ or its variants), which will be used for our study 
in {Section \ref{sec:app} below}.

\subsubsection{Homogenization of polynomials and monomial orderings}\label{sec:homogenization}
We begin with recalling the notion of homogenization. (See \cite[Chapter 4]{KR} 
for details.)
Let $K$ be a field, $X=\{x_1,\ldots,x_n\}$ a set of variables, and
$\mathcal{T}$ the set of all monomials in $X$. 
\footnote{As the \textcolor{black}{symbol} $m$ is used for the size of a generating set, 
we use $\mathcal{T}$ instead of $\mathcal{M}$.}
\begin{enumerate}
\item[(1)]
For a non-homogeneous polynomial $f=\sum_{{t}\in \mathcal{T}} c_t {t}$
in $K[X]$ {with $c_t \in K$},
its {\em homogenization} $f^h$ is defined, by introducing a new
variable $y$, as
\[
f^h=\sum_{{t}\in \mathcal{T}} c_t {t} {y^{\deg(f)-{\deg(t)}}}.
\]
Thus $f^h$ is a homogeneous polynomial in $X\cup \{y\}$ over $K$ with total degree {\em $d=\deg(f)$}.
Also for a set $F$ (or a sequence $\bm{F} = (f_1, \ldots , f_m)\in K[X]^m$) of polynomials, its {\em homogenization}
$F^h$ (or $\bm{F}^h$) is defined as $F^h=\{f^h\;|\; f\in F\}$ (or $\bm{F}^h = (f_1^h,\ldots , f_m^h)\in K[X \cup \{ y\}]^m$).
We also write $X^h$ for $X\cup\{y\}$.
\item[{\rm (2)}]
Conversely, for a homogeneous polynomial $h$ in $K[X\cup\{y\}]$,
its {\em dehomogenization} $h^{\rm deh}$ is defined
by substituting $y$ with $1$, that is, $h^{\rm deh}=h(X,1)$.
(It is also denoted by $h|_{y=1}$.)
For a set $H$ of homogeneous polynomials in $K[X\cup \{y\}]$,
its {\em dehomogenization} $H^{\rm deh}$ (or $H|_{y=1}$) is defined
as $H^{\rm deh}=\{h^{\rm deh}\;|\;h\in H\}$.
\textcolor{black}{We also apply the dehomogenization to sequences of polynomials.}
\item[{\rm (3)}]
For an ideal $I$ of $K[X]$,
its homogenization $I^h$, as an ideal,
is defined as $\langle I^h\rangle_{K[X\cup\{y\}]}$.
\textcolor{black}{We remark that, for a set $F$ of polynomials in $K[X]$, we have $\langle F^h \rangle_{K[X^h]} \subset I^h$ with $I = \langle F \rangle_{K[X]}$, and the equality does not hold in general.}

\item[(4)]
For a homogeneous ideal $J$ in $K[X\cup\{y\}]$,
its dehomogenization $J^{\rm deh}$, as a set, is an ideal of $K[X]$.
We note that if a homogeneous ideal $J$ is generated by $H$, then
$J^{\rm deh}=\langle H^{\rm deh}\rangle_{K[X]}$ and for an ideal $I$ of $K[X]$,
{we have} $(I^h)^{\rm deh}=I$. 
\item[{\rm (5)} ]
For a monomial (term) ordering $\prec$ on the set of {\em monomials}
$\mathcal{T}$ in $X$,  
its {\em homogenization} $\prec_h$ on the set of {\em monomials} 
$\mathcal{T}^h$ in $X\cup \{y\}$ is defined as follows:
For two monomials $X^\alpha y^a$ and $X^\beta y^b$ in $\mathcal{T}^h$,
we say $X^\alpha y^a \prec_h X^\beta y^b$ if and only if one of the following holds:
(i) $a+|\alpha| <  b+|\beta|$, or (ii) $a+|\alpha|= b+|\beta|$ and $X^\alpha\prec X^\beta$,
where $\alpha=(\alpha_1,\ldots,\alpha_n)\in \mathbb{Z}_{\geq 0}^n$ 
and , 
{$|\alpha|=\alpha_1+\cdots + \alpha_n$}, 
and where $X^\alpha$ denotes 
{$x_1^{\alpha_1}\cdots x_n^{\alpha_n}$}.
{ 
Here, for a monomial $X^\alpha {y^a}$, we call $X^\alpha$ and ${y^a}$ its {\em $X$-part} and its {\rm $\{y\}$-part} (or $y$-part simply), respectively. }
If a monomial ordering $\prec$ is {\em graded}, 
the restriction $\prec_h|_{\mathcal{T}}$
of $\prec_h$ on $\mathcal{T}$ coincides with $\prec$.
\end{enumerate}

It is well-known that 
for a Gr\"obner basis $H$ of $\langle F^h\rangle$ with respect to $\prec^h$, 
its \textcolor{black}{dehomogenization} $\{h^{\rm deh}\;|\; h\in H\}$ is also a Gr\"obner basis of $\langle F\rangle$ 
with respect to $\prec$. 

\subsubsection{Signature and $F_5$ algorithm}\label{sec:sba}

Here we briefly outline the $F_5$ algorithm, 
which is an improvement of Buchberger's algorithm. 
For details, see a survey \cite{EF}. 
Let $F=\{f_1,\ldots,f_m\}\subset R=K[X]$ be a given generating set. 
For each polynomial $h$ constructed during the Gr\"obner basis 
computation of $\langle F\rangle$, 
the $F_5$ algorithm 
attaches a {\em special label called {a} signature} as follows:
Since $h$ belongs to $\langle F\rangle$, it can be written as 
\begin{equation}\label{eq:sign}
h=a_1f_1+a_2f_2+\cdots + a_m f_m
\end{equation}
for some $a_1,\ldots,a_m\in R$. 
Then, we assign $h$ to $a_1\mathbf{e}_1+\cdots +a_m\mathbf{e}_m\in R^m$ and 
we call its leading monomial $t\mathbf{e}_i$ with respect to 
a monomial (module) ordering in $R^m$ {\em the signature} of $h$. 
As the expression \eqref{eq:sign} is not unique, in order to 
determine the signature, we construct the expression 
procedurally or use the uniquely determined residue in $R^m/{\rm syz}(f_1,
\ldots,f_m)$ by a module Gr\"obner basis of ${\rm syz}(f_1,\ldots,f_m)$. 
(For the latter case, we call it the {\em minimal} signature.)
Here we denote the signature of $h$ by ${\rm sig}(h)$. 
Anyway, in the $F_5$ algorithm, we can meet the both by 
carefully choosing S-polynomials and 
by applying restricted reduction steps (called $\Sigma$-reductions) 
for S-polynomials without any change of the signature. (So, we need not compute a module 
Gr\"obner basis of ${\rm syz}(f_1,\ldots,f_m)$.) 
We note that for the S-polynomial $S(h_1,h_2)={c_1 t_1 h_1-c_2 t_2 h_2}$ {with $c_1,c_2 \in K$ and $t_1,t_2 \in \mathcal{T}$}, 
the signature {${\rm sig}(S(h_1,h_2))$} 
is determined as the largest one 
between ${\rm sig}({c_1t_1}h_1)$ and ${\rm sig}({c_2t_2}h_2)$. 
Then, we have the following criteria, which are very useful 
to avoid the computation of unnecessary S-polynomials. (The latter 
one is called the {\em syzygy criterion}.) 

\begin{proposition}[cf.\ \cite{CLO}, \cite{EF}]\label{prop:syz}
In the $F_5$ algorithm, we need {not} compute an S-polynomial 
if some S-polynomial of the same signature was already proceeded, 
since both are reduced to the same polynomial. 
Moreover, we need {not} compute an S-polynomial of signature $s$
if there is a signature $s'$ such that $s'$ divides $s$ and 
some S-polynomial with the signature $s'$ is reduced to $0$. 
\end{proposition}

\subsection{Complexity of the Gr\"{o}bner basis computation}\label{subsec:complexity}

In general, determining the complexity of computing a Gr\"{o}bner basis is very hard;
in the worst-case, the complexity is doubly exponential in the number of variables, see e.g., \cite{C14}, \cite{MR13}, \cite{Ritscher} for surveys.
It is well-known that a Gr\"{o}bner basis with respect to a graded monomial ordering (in particular, DRL ordering) can be computed quite more efficiently than ones with respect to other orderings in general. 
Moreover, in the case where the input polynomials generate a zero-dimensional ideal, once a Gr\"{o}bner basis with respect to an efficient monomial ordering 
is computed, one with respect to any other ordering can be computed easily by the FGLM basis conversion algorithm~\cite{FGLM}.
From this, we focus on the case where the monomial ordering is graded, and if necessary we also assume that the ideal generated by the input polynomials is zero-dimensional.

A typical way to estimate the complexity \textcolor{black}{of computing a Gr\"{o}bner basis for a sequence $\bm{F}$ of polynomials} is to count {the number of} S-polynomials that are reduced during the Gr\"{o}bner basis computation.
In the case where the chosen monomial ordering is graded, the most efficient strategy to compute Gr\"{o}bner bases is the normal strategy, on which we proceed {\it degree by degree}, namely increase the degree of critical pairs defining S-polynomials, as in the $F_4$ and $F_5$ algorithms.
For an algorithm adopting this strategy, several S-polynomials are dealt with 
{consecutively at} the same degree, which is called the {\it step degree}.
The \textcolor{black}{highest} step degree at which an intermediate ideal basis contains a minimal Gr\"{o}bner basis is called the {\it solving degree} of the algorithm, and it is denoted by $\mathrm{sd}_{\prec}^{\rm hsd}(\bm{F})$.
Determining (or finding a tight bound for) the solving degree is difficult without computing any Gr\"{o}bner basis.
\textcolor{black}{
Once it is specified, we may estimate the complexity of the algorithm, as in \cite{Tenti}.}

On the other hand, for a linear algebra-based algorithm, such as an $F_4$-family including the (matrix-)$F_5$ algorithm and the XL family {(cf.\ \cite{XL})}, that follows Lazard's \textcolor{black}{strategy}~\cite{Lazard} to reduces S-polynomials by the Gaussian elimination on Macaulay matrices, \textcolor{black}{Caminata-Gorla~\cite{CG20} defined {\it another solving degree} in a different manner.
Specifically, it is defined as the lowest degree $d$ at which the reduced row echelon form (RREF) of the Macaulay matrix $M_{\leq d}(\bm{F})$ produces a Gr\"{o}bner basis, see \cite{CG20} for details.
In this case,} the complexity is estimated to be $O(N^{\omega})$ with $N = \binom{n+d}{n}$, where $\omega$ is the {\it matrix multiplication exponent} with $2 \leq \omega < 3$.
For a given polynomial sequence $\bm{F}=(f_1,\ldots , f_m) \in R^m$ and a graded monomial ordering $\prec$, we denote by $\mathrm{sd}_{\prec}^{\rm mac}(\bm{F})$ this solving degree.
\textcolor{black}{
In a series of works (cf.\ \cite{CG20}, \cite{BNDGMT21}, \cite{CG23}) by Gorla et al.,} they provided a mathematical formulation for the relation between the solving degree $\mathrm{sd}_{\prec}^{\rm mac}(\bm{F})$ (or $\mathrm{sd}_{\prec}^{\rm mut}(\bm{F})$ described below) and algebraic invariants coming from $\bm{F}$, such as the maximal Gr\"{o}bner basis degree, the degree of regularity, the Castelnuovo–Mumford regularity, the first and last fall degrees, and so on.
Here, the {\it maximal Gr\"{o}bner basis degree} of the ideal $\langle \bm{F} \rangle_R$ is the maximal degree of elements in the reduced Gr\"{o}bner basis of $\langle \bm{F} \rangle_R$ with respect to a fixed monomial ordering $\prec$, and is denoted by $\mathrm{max.GB.deg}_{\prec}(\bm{F})$.

In the following, we recall some of Caminata et al.'s results.
We set $\prec$ as the DRL ordering on $R$ with $x_n \prec \cdots \prec x_1$, and fix it throughout the rest of this subsection.
Let $y$ be an extra variable for {homogenization} as in the previous subsection, and $\prec^h$ the homogenization of $\prec$, so that $ y \prec x_i$ for any $i$ with $1 \leq i \leq n$.
Then, we have
\[
{\rm max.GB.deg}_{\prec}(\bm{F}) \leq \mathrm{\rm sd}_{\prec}^{\rm mac}(\bm{F}) =\mathrm{\rm sd}_{\prec^h}^{\rm mac}(\bm{F}^h) = {\rm max.GB.deg}_{\prec^h}(\bm{F}^h),
\]
see \cite{CG20} for a proof.
Here, we also recall Lazard's bound for the maximal Gr\"{o}bner basis degree of $\langle \bm{F}^h \rangle_{R'}$ with $R' = R[y]$:

\begin{theorem}[Lazard; {\cite[Theorem 2]{Lazard}}]\label{thm:Lazard_zero_dim}
With notation as above, we assume that the number of projective zeros of $\bm{F}^h$ is finite (and therefore $m \geq n$), and that $f_1^h=\cdots = f_m^h=0$ has no non-trivial solution over the algebraic closure $\overline{K}$ with $y=0$, i.e., $\bm{F}^{\rm top}$ has no solution in $\overline{K}^n$ other than $(0,\ldots , 0)$.
Then, supposing also that $d_1 \geq \cdots \geq d_m$ and putting $\ell := \mathrm{min}\{ m,n+1 \}$, we have
\begin{equation}\label{eq:Lazard0}
\mathrm{max.GB.deg}_{\prec^h}(\bm{F}^h) \leq d_1 + \cdots + d_{\ell} - \ell + 1
\end{equation}
\end{theorem}

{Lazard's bound} given in \eqref{eq:Lazard0} is also referred to as the {\it Macaulay bound}, and it provides an upper-bound for the solving degree of $\bm{F}$ with respect to {a DRL ordering}.

As for the maximal Gr\"{o}bner basis degree of $\langle \bm{F} \rangle$, if $\langle \bm{F}^{\rm top} \rangle$ is \textcolor{black}{Artinian}, we have
\begin{equation*}
\mathrm{max.GB.deg}_{\prec'}(\bm{F}) \leq d_{\rm reg}(\langle \bm{F}^{\rm top} \rangle)
\end{equation*}
for any graded ordering $\prec'$ on $R$, see \cite[Remark 15]{CG20} or Lemma \ref{lem:sd} below for a proof.
{Both $d_{\rm reg}(\langle \bm{F}^{\rm top} \rangle)$ and $\mathrm{\rm sd}_{\prec}^{\rm mac}(\bm{F})$ are greater than or equal to ${\rm max.GB.deg}_{\prec}(\bm{F})$, whereas the degree of regularity (or the first fall degree) used in the cryptographic literature as a proxy (or a heuristic upper-bound) for the solving degree.}
However, it is pointed out in \cite{BNDGMT21}, \cite{CG20}, and \cite{CG23} by explicit examples that {\it any} of the degree of regularity and the first fall degree does {\it not} produce an estimate for the solving degree in general, even when $\bm{F}$ is an affine (cryptographic) semi-regular sequence.
\textcolor{black}{
In \cite{CG23}, Caminata-Gorla provided yet another solving degree, denoted by $\mathrm{sd}_{\prec'}^{\rm mut}(\bm{F})$, with respect to algorithms based on the {\it mutant strategy} (see \cite{CG23-2} for details), and they proved} that it is nothing but the {\it last fall degree} if it is greater than the maximal Gr\"{o}bner basis degree:

\begin{theorem}[{\cite[Theorem 3.1]{CG23}}]\label{thm:sd}
    With notation as above, for any graded monomial ordering $\prec'$ on $R$, we have the following inequality:
    \[
    \mathrm{sd}_{\prec'}^{\rm mut}(\bm{F})  = \max \{ d_{\bm{F}}, {\rm max.GB.deg}_{{\prec'}}(\bm{F}) \},
    \]
    where $d_{\bm{F}}$ denotes the last fall degree of $\bm{F}$ defined in \cite[Definition 1.5]{CG23}.
\end{theorem}

By this theorem, if $ d_{\rm reg}(\langle \bm{F}^{\rm top} \rangle) < d_{\bm{F}}$, the degree of regularity is no longer an upper-bound on the solving degree.

On the other hand, Semaev and Tenti \textcolor{black}{claimed (see Tenti's thesis~\cite{Tenti} for a proof)} that the solving degree $\mathrm{sd}_{\prec}^{\rm hsd}(\bm{F})$ \textcolor{black}{(in terms of the highest step degree)} is linear in the degree of regularity, if $K$ is a (large) finite field, and if the input system 
contains {polynomials related to} 
the {\it field equations}, say {$x_i^q - x_i$} for $1 \leq i \leq n$:

\begin{theorem}[{\cite[Theorem 2.1]{ST}, \cite[Corollary 3.67]{Tenti}}]\label{thm:tenti}
    With notation as above, assume that $K=\mathbb{F}_q$, and that $\bm{F}$ contains $x_i^q - x_i$ for $1 \leq i \leq n$.
    Put $D=d_{\rm reg}(\langle \bm{F}^{\rm top} \rangle)$.
    If $ D \geq \max \{ \deg (f) : f \in \bm{F} \}$ {and $D \geq q$}, then we have
    \begin{equation}\label{eq:tenti}
            \mathrm{sd}_{\prec}^{\rm hsd}(\bm{F}) \leq 2 D  - 2. 
    \end{equation}
\end{theorem}

In Subsection \ref{subsec:overD} below, we will prove a similar inequality \eqref{eq:tenti} for the case where $\bm{F}$ not necessarily contains a field equation but is cryptographic semi-regular.

\section{Hilbert-Poincar\'{e} series of affine semi-regular sequence}\label{sec:main}

As in the previous section, let $K$ be a field, and $R = K[X] = K[x_1, \ldots , x_n]$ denote the polynomial ring of $n$ variables over $K$.
We denote by $R_{d}$ the homogeneous part of degree $d$, that is, the set of homogeneous polynomials of degree $d$ and $0$.
Recall Definition \ref{def:csemireg} for the definition of cryptographic semi-regular sequences.

The Hilbert-Poincar\'{e} series associated to a (homogeneous) cryptographic semi-regular sequence is given by \eqref{eq:semiregHil2}.
On the other hand, the Hilbert-Poincar\'{e} series \textcolor{black}{associated to} the 
\textcolor{black}{homogenization} $\bm{F}^h$ of $\bm{F}=(f_1,\ldots,f_m)\in R^m$ \textcolor{black}{not necessarily homogeneous polynomials} {cannot be} computed 
without knowing its Gr\"obner basis in general, but we shall prove that it can be computed up to the degree $d_{\rm reg}(\langle \bm{F}^{\rm top}\rangle)$ 
if $\bm{F}$ is affine cryptographic semi-regular, namely $\bm{F}^{\rm top}$ is cryptographic semi-regular.

\begin{theorem}\label{thm:main}
Let $R = K[x_1,\ldots , x_n]$ and $R' = R[y]$, and let $\bm{F}=(f_1,\ldots , f_m)$ be a sequence of not necessarily homogeneous polynomials in $R$.
Assume that $\bm{F}$ is affine cryptographic semi-regular.
Then, for each $d$ with $d < D:= d_{\mathrm{reg}}(\langle \bm{F}^{\rm top} \rangle)$, we have
    \begin{equation}\label{eq:main}
            {\rm HF}_{R'/\langle \bm{F}^h \rangle}(d) = {\rm HF}_{R/\langle \bm{F}^{\rm top} \rangle}(d) + {\rm HF}_{{R'}/\langle \bm{F}^h \rangle}(d-1),
    \end{equation}
    and hence
    \begin{equation}\label{eq:main2}
            {\rm HF}_{R'/\langle \bm{F}^h \rangle}(d) = {\rm HF}_{R/\langle \bm{F}^{\rm top} \rangle}(d) + \cdots + {\rm HF}_{R/\langle \bm{F}^{\rm top} \rangle}(0),
    \end{equation}
    whence we can compute the value ${\rm HF}_{R'/\langle \bm{F}^h \rangle}(d)$ from the formula \eqref{eq:semiregHil2}.
\end{theorem}

\begin{proof}
Let $K_{\bullet} = K_{\bullet}(f_1^h,\ldots,f_m^h)$ be the Koszul complex on $(f_1^h,\ldots , f_m^h)$, which is given by \eqref{eq:Koszul}.
By tensoring $K_{\bullet}$ with $R'/\langle y \rangle_{R'} \cong K[x_1, \ldots , x_{n}]=R$ over $R'$, we obtain the following exact sequence of chain complexes:
\[{
	\xymatrix{
	 0 \ar[r]	& K_{\bullet} \ar[r]^{\times y} & K_{\bullet} \ar[r]^(0.4){\pi_{\bullet}} & K_{\bullet} \otimes_{R'} R \ar[r] & 0 ,
	}}
\]
where $\times y$ is a graded homomorphism of degree $1$ multiplying each entry of a vector with $y$, \textcolor{black}{and where $\pi_i$ is a canonical homomorphism sending $v \in K_i$ to $v_i \otimes 1 \in K_i \otimes_{R'} R$.}
Note that there is an isomorphism
\[
K_i \otimes_{R'} R \cong \bigoplus_{1 \leq j_1 < \cdots < j_i \leq m} R(-d_{j_1\cdots j_i}) \mathbf{e}_{j_1 \cdots {{j_i}}},
\]
\textcolor{black}{via which we can interpret $\pi_i:K_i \to K_i \otimes_{R'} R$ as a homomorphism that projects each entry of a vector in $K_i$ modulo $y$.}
In particular, we have
\begin{eqnarray}
K_0 \otimes_{R'} R &=& R' / \langle f_1^h, \ldots , f_m^h \rangle_{R'} \otimes_{R'} R'/\langle y \rangle_{R'} \nonumber \\
& \cong & R' / \langle f_1^h, \ldots , f_m^h , y \rangle_{R'} \nonumber \\
& \cong & R / \langle f_1^{\rm top}, \ldots , f_m^{\rm top} \rangle_R \nonumber
\end{eqnarray}
for $i=0$.
This means that the chain complex $K_{\bullet} \otimes_{R'} R$ gives rise to the 
\textcolor{black}{Koszul} complex on $({f}_1^{\rm top},\ldots,f_m^{\rm top})$.
We induce a long exact sequence of homology groups.
In particular, for each degree $d$, we have the following long exact sequence:
\[{
	\xymatrix{
	  H_{i+1}(K_{\bullet})_{d-1} \ar[r]^{\times y} & H_{i+1}(K_{\bullet})_d \ar[r]^(0.4){\pi_{i+1}} & H_{i+1}(K_{\bullet} \otimes_{R'} R)_d \ar[lld]_{\delta_{i+1}} \\
	  H_{i}(K_{\bullet})_{d-1} \ar[r]_{\times y} & H_{i}(K_{\bullet})_d \ar[r]_(0.4){\pi_i} & H_{i}(K_{\bullet}\otimes_{R'} R)_d ,    \\
	}}
\]
where $\delta_{i+1}$ is the connecting homomorphism produced by the Snake lemma.
For $i=0$, we have the following exact sequence:
\[
H_{1}(K_{\bullet} \otimes_{R'} R)_d \longrightarrow H_{0}(K_{\bullet})_{d-1} \overset{\times y}{\longrightarrow} H_{0}(K_{\bullet})_d \longrightarrow H_{0}(K_{\bullet}\otimes_{R'} R)_d \longrightarrow 0. 
\]
From our assumption that $\bm{F}^{\rm top}$ is cryptographic semi-regular, it follows from \textcolor{black}{Proposition} \ref{prop:Diem} that $H_1(K_{\bullet} \otimes_{R'}R)_{\leq D-1} = 0$ for $D := d_{\rm reg}(\langle \bm{F}^{\rm top} \rangle)$.
Therefore, if $d \leq D-1$, we have an exact sequence
\[{
	\xymatrix{
	  0 \ar[r] & H_{0}(K_{\bullet})_{d-1} \ar[r]^{\times y} & H_{0}(K_{\bullet})_d \ar[r] & H_{0}(K_{\bullet}\otimes_{R'} R)_d \ar[r] & 0    \\
	}}
\]
of $K$-linear spaces, so that
 \[
    \mathrm{dim}_K H_{0}(K_{\bullet})_d = \mathrm{dim}_K H_{0}(K_{\bullet}\otimes_{R'} R)_d + \mathrm{dim}_K H_{0}(K_{\bullet})_{d-1}
\]
by the dimension theorem.
Since $H_0(K_{\bullet}) = R'/ \langle \bm{F}^h \rangle$ and ${H_0 (K_{\bullet}\otimes_{R'} R)} \cong R / \langle \bm{F}^{\rm top} \rangle$, we have the equality \eqref{eq:main}, as desired.
\end{proof}

\begin{remark}
\textcolor{black}{With notation as in Theorem \ref{thm:main}, assume that $D < \infty$ (and thus $m \geq n$).}
{
In the proof of Theorem \ref{thm:main}, the multiplication map $H_0(K_{\bullet})_{d-1} \to H_0(K_{\bullet})_{d}$ by $y$ is injective for all $d < D$, whence ${\rm HF}_{R'/\langle F^h \rangle}(d)$ is monotonically increasing for $d < D-1$.
On the other hand, since \textcolor{black}{$H_{0}(K_{\bullet}\otimes_{R'} R)_{d} = (R/\langle F^{\rm top}\rangle)_d =0$} for all $d \geq D$ by the definition of the degree of regularity, the multiplication map $H_0(K_{\bullet})_{d-1} \to H_0(K_{\bullet})_{d}$ by $y$ is surjective for all $d \geq D$, whence ${\rm HF}_{R'/\langle F^h \rangle}(d)$ is monotonically decreasing for $d \geq D-1$.
By this together with \cite[Theorem 3.3.4]{C14}, the homogeneous ideal $\langle \bm{F}^h \rangle$ is zero-dimensional or trivial, i.e., there are at most a finite number of projective zeros of $F^h$ (and thus there are at most a finite number of affine zeros of $\bm{F}$).}
\end{remark}

By Theorem 4, it can be proved that the Hilbert-Poincar\'e series of $R'/
\langle \bm{F}^h\rangle$ satisfies the following equality~\eqref{eq:hHilb}, which 
may correspond to \cite[Proposition 6]{BFSY}: 

\begin{corollary}\label{cor:Dreg}
Let $D=d_{\rm reg}(\langle \bm{F}^{\rm top}\rangle)$. Then we have
\begin{equation}\label{eq:hHilb}
{\rm HS}_{R'/\langle \bm{F}^h\rangle}(z)\equiv \frac{\prod_{i=1}^m (1-z^{d_i})}{(1-z)^{n+1}}\pmod{z^D}.
\end{equation}
Therefore, by Theorem \ref{lem:Diem2} (\cite[Theorem 1]{Diem2}), 
$\bm{F}^h$ is $D$-regular. 
{Here, we note that $D=\deg({\rm HS}_{R/\langle \bm{F}^{\rm top}\rangle}{ )}+1
={\deg\left(\left[\frac{\prod_{i=1}^m(1-z^{d_i})}{(1-z)^n}\right]\right)}+1$. }
\end{corollary}
\begin{proof} Let ${\rm HS}'(z)=\frac{\prod_{i=1}^m (1-z^{d_i})}{(1-z)^{n+1}}\bmod{z^D}$ and let
${\rm HF}'(d)$ denote the coefficient of ${\rm HS'}(z)$ of degree $d$ for $d<D$. 
First we remark that, 
as $\bm{F}^{\rm top}$ is a cryptographic semi-regular sequence, 
the Hilbert-Poincar\'e series of $R/\langle \bm{F}^{\rm top}\rangle$ satisfies the 
following:
\[
{\rm HS}_{R/\langle \bm{F}^{\rm top}\rangle}(d) 
=\left[\frac{\prod_{i=1}^m (1-z^{d_i})}{(1-z)^n}\right]
=\frac{\prod_{i=1}^m (1-z^{d_i})}{(1-z)^n}\bmod{z^D},
\]
{since} ${\rm HF}_{R/\langle \bm{F}^{\rm top}\rangle}(d)=0$ 
for $d\geq D$. 
Then we have
\begin{eqnarray*}
{\rm HS}'(z) \bmod{z^D} & = & 
\frac{\prod_{i=1}^m (1-z^{d_i})}{(1-z)^{n+1}}\bmod{z^D}\\
& = & \frac{\prod_{i=1}^m (1-z^{d_i})}{(1-z)^n}\times (1+z+\cdots+z^{D-1})
\bmod{z^D}\\
& = & {\rm HS}_{R/\langle \bm{F}^{\rm top}\rangle}(z) \cdot (1+z+\cdots +z^{D-1})\bmod{z^D}. 
\end{eqnarray*}
Therefore, for $d<D$, the equation \eqref{eq:main2} gives 
\[
{\rm HF}'(d)={\rm HF}_{R/\langle \bm{F}^{\rm top}\rangle}(d)+\cdots +
{\rm HF}_{R/\langle \bm{F}^{\rm top}\rangle}(0)={\rm HF}_{R'/\langle \bm{F}^h\rangle}(d),
\]
{which implies the desired equality \eqref{eq:hHilb}.}
\end{proof}

\textcolor{black}{To prove the following corollary}, we use a fact that, for a homogeneous ideal $I$ in $R$, the equality $\sum_{i=0}^d \mathrm{dim}_K I_i = \mathrm{dim}_K (IR')_d$ holds for each $d \geq 0$.   
\textcolor{black}{Also we take a graded ordering $\prec$ (preferably a DRL ordering) on monomials in $X$ and its homogenization on monomials in $X\cup\{y\}$. }

\begin{corollary}\label{cor:LM}
With notation as above, assume that $\bm{F} = (f_1,\ldots , f_m) \in R^m$ is affine cryptographic semi-regular.
Put $\overline{I} := \langle \bm{F}^{\rm top} \rangle_R$ and $\tilde{I} := \langle \bm{F}^h \rangle_{R'}$.
Then, we have $(\langle \mathrm{LM}(\tilde{I}) \rangle_{R'})_d = (\langle \mathrm{LM}(\overline{I}) \rangle_{R'})_d $ for each $d$ with $d < D:=d_{\rm reg} (\overline{I})$.
\end{corollary}

\begin{proof}
We prove $(\langle \mathrm{LM}(\tilde{I}) \rangle_{R'})_d \subset (\langle \mathrm{LM}(\overline{I}) \rangle_{R'})_d $ by the induction on $d$.
The case where $ d =0$ is clear from Theorem \ref{thm:main}, and so we assume $d >0$.
{Any element in $(\langle \mathrm{LM}(\tilde{I}) \rangle_{R'})_d$ is represented as a finite sum of elements in $R'$ of the form $g \cdot \mathrm{LM}(h)$ with $g \in R'$, $h \in \tilde{I}$, and $\mathrm{deg}(gh) = d$.
Hence, we can also write each $g \cdot \mathrm{LM}(h)$ as a $K$-linear combination of elements of the form $\mathrm{LM}(t h)$ for a monomial $t$ in $R'$ of 
\textcolor{black}{degree} $\mathrm{deg}(g)$, where $t h$ is an element in $\tilde{I}$ of degree $d$. 
Therefore, it suffices for showing ``$\subset$'' to prove that $\mathrm{LM}(f) \in (\langle \mathrm{LM}(\overline{I}) \rangle_{R'})_d$ for any $f \in \tilde{I}$ with $\mathrm{deg}(f) = d$.}
We may assume that $f$ is homogeneous.
It is straightforward that $f|_{y=0} \in \overline{I}_{\leq d}$.
If $\mathrm{LM}(f) \in R=K[x_1,\ldots , x_{n}]$, then we have $\mathrm{LM}(f) = {\rm LM}(f|_{y=0}) \in \mathrm{LM}(\overline{I})$.
Thus, we may also assume that $y \mid \mathrm{LM}(f)$.
In this case, it follows from the definition of the DRL ordering that any other term in $f$ is also divisible by $y$, so that $f \in \langle y \rangle_{R'}$.
Thus, we can write $f = y h$ for some $h \in R'$, where $h$ is homogeneous of degree $d-1$.
As in the proof of Theorem \ref{thm:main}, the multiplication map
\[
(R'/\tilde{I})_{d'-1} \to (R'/\tilde{I})_{d'} \ ; \ h' \bmod{\tilde{I}} \mapsto  y h' \bmod{\tilde{I}}
\]
is injective for any $d' < d_\mathrm{reg}( \overline{I} )$, since $F$ is cryptographic semi-regular. 
Therefore, it follows from $f \in \tilde{I}_d$ that $h \in \tilde{I}_{d-1}$, whence $f = y h \in y \tilde{I}_{d-1}$.
By the induction hypothesis, there exists $g\in \overline{I}$ such that $\mathrm{LM}(g) \mid \mathrm{LM}(h)$, whence $\mathrm{LM}(f) \in (\langle \mathrm{LM}(\overline{I}) \rangle_{R'})_d$.

Here, it follows from Theorem \ref{thm:main} that
\begin{eqnarray*}
    \mathrm{dim}_K (R')_d - \mathrm{dim}_K \tilde{I}_d & = & \sum_{i=0}^d \left( \mathrm{dim}_K R_i -  \mathrm{dim}_K {\overline{I}}_i \right) = \sum_{i=0}^d \mathrm{dim}_K R_i - \sum_{i=0}^d  \mathrm{dim}_K {\overline{I}}_i \\
    &=& \mathrm{dim}_K (R')_d - \mathrm{dim}_K (\overline{I} R')_d,
\end{eqnarray*}
and thus $\mathrm{dim}_K \tilde{I}_d = \mathrm{dim}_{K} (\overline{I} R')_d$.
Hence, it follows from $\langle \mathrm{LM}(\overline{I})\rangle_{R'} = \langle \mathrm{LM}(\overline{I}R') \rangle_{R'}$ that
\[
\mathrm{dim}_K (\langle \mathrm{LM}(\tilde{I}) \rangle_{R'})_d = \mathrm{dim}_K (\langle \mathrm{LM}(\overline{I}) \rangle_{R'})_d,
\]
whence $(\langle \mathrm{LM}(\tilde{I}) \rangle_{R'})_d = (\langle \mathrm{LM}(\overline{I}) \rangle_{R'})_d $, as desired.
\end{proof}

\begin{example} We give a simple example. Let $p=73$, $K={\mathbb F}_p$, and 
\begin{eqnarray*}
f_1 & = & x_1^2+(3x_2-2x_3-1)x_1+x_2^2+(-2x_3-2)x_2+x_3^2+x_3,\\
f_2 & = & 4x_1^2+(3x_2+4x_3-2)x_1-x_2+x_3^2+2x_3,\\
f_3 & = & 3x_1^2-x_1+9x_2^2+(-6x_3+1)x_2+x_3^2-x_3,\\
f_4 & = & x_1^2+(-6x_2+2x_3-2)x_1+9x_2^2+(-6x_3+1)x_2+2x_3^2.
\end{eqnarray*}
Then, $d_1=d_2=d_3=d_4=2$. As their {top parts (maximal total degree parts)} are 
\begin{eqnarray*}
f_1^{\rm top} & = & x_1^2+(3x_2-2x_3)x_1+x_2^2-2x_3x_2+x_3^2,\\
f_2^{\rm top} & = & 4x_1^2+(3x_2+4x_3)x_1+x_3^2,\\
f_3^{\rm top} & = & 3x_1^2+9x_2^2-6x_3x_2+x_3^2,\\
f_4^{\rm top} & = & x_1^2+(-6x_2+2x_3)x_1+9x_2^2-6x_3x_2+2x_3^2, 
\end{eqnarray*}
{one can verify that} $\bm{F}^{\rm top}$ is a cryptographic semi-regular sequence. 
Moreover, its degree of regularity is equal to $3$. 
Indeed, the reduced Gr\"obner basis $G_{\rm top}$ of 
the ideal $\langle \bm{F}^{\rm top}\rangle$ with respect to {the DRL ordering} $x_1\succ x_2\succ x_3$ 
is 
\[
\{x_3^2x_2,x_3^3,x_1^2+68x_3x_2+55x_3^2,x_2x_1+27x_3x_2+29x_3^2,x_2^2+x_3x_2+71x_3^2,x_3x_1+3x_3x_2+33x_3^2\}.
\]
Then its leading monomials are $x_3^3,x_3^2x_2,x_1^2,x_1x_2,x_2^2,x_3x_1$ and its Hilbert-Poincar\'e series 
satisfies 
\[
{\rm HS}_{R/\langle \bm{F}^{\rm top}\rangle}(z) 
=2z^2+3z+1 
=\left( \frac{(1-z^2)^4}{(1-z)^3}\bmod{z^3} \right),
\]
whence the degree of regularity of $\langle \bm{F}^{\rm top} \rangle$ is 3. 

On the other hand, 
the reduced Gr\"obner basis $G_{\rm hom}$ of the ideal 
$\langle \bm{F}^h\rangle$ with respect to {the DRL ordering} 
$x_1\succ x_2\succ x_3\succ y$  is 
\begin{align*}
& \{y^3x_1,y^3x_2,y^3x_3,60y^2x_1+(x_3^2+22y^2)x_2+39y^2x_3,\\
& 72y^2x_1+14y^2x_2+x_3^3+56y^2x_3,
16y^2x_1+(yx_3+55y^2)x_2+38y^2x_3,\\
& 72y^2x_1+66y^2x_2+yx_3^2+70y^2x_3,
x_1^2+72yx_1+(68x_3+40y)x_2+55x_3^2+14yx_3,\\
& (x_2+20y)x_1+(27x_3+37y)x_2+29x_3^2+12yx_3,\\
& 57yx_1+x_2^2+(x_3+3y)x_2+71x_3^2+52yx_3,\\
& (x_3+22y)x_1+(3x_3+5y)x_2+33x_3^2+14yx_3\}
\end{align*}
and its leading monomials are $y^3x_1$, $y^3x_2$, $y^3x_3$, $x_3^2x_2$, $x_3^3$, $yx_2x_3$, 
$yx_3^2$, $x_1^2$, $x_1x_2$, $x_2^2$, and $x_1x_3$. 
Then the Hilbert-Poincar\'e series of $R'/\langle \bm{F}^h\rangle$ satisfies 
\[
\left( {\rm HS}_{R'/\langle \bm{F}^{h}\rangle}(z) \bmod{z^3} \right) =\left( 6z^2+4z+1 \bmod{z^3} \right)= \left( \frac{(1-z^2)^4}{(1-z)^4}\bmod{z^3} \right).
\]
We note that ${\rm HF}_{R'/\langle \bm{F}^h\rangle}(3)=4$ and ${\rm HF}_{R'/\langle \bm{F}^h\rangle}(4)=1$. 
We can also examine ${\rm LM}(G_{\rm hom})_{d<D}={\rm LM}(G_{\rm top})_{d<D}$ 
and, for $g\in G_{\rm hom}$, if ${\rm LM}(g)$ is divided by $y$, then 
$\deg(g)\geq D=3$. Thus, at the degree 3, there occurs a 
{\em degree-fall}. \textcolor{black}{See \cite[Sebsection 2.1]{CG23} for details.} 
Also, the reduced Gr\"{o}bner basis of $\langle \bm{F}\rangle$ 
with respect to $\prec$ 
is $\{x,y,z\}$ and we can examine that 
the \textcolor{black}{dehomogenization} of $G_{\rm hom}$ is also a Gr\"obner basis 
of $\langle \bm{F}\rangle$. 
\end{example}

\section{Application to Gr\"obner bases computation}\label{sec:app}

{We use the same notation as in the previous section, and assume that $\bm{F}$ is cryptographic semi-regular \textcolor{black}{such that $D := d_{\mathrm{reg}}(\langle \bm{F}^{\rm top} \rangle)<\infty$}.}
Here we apply results in the previous section to the computation of Gr\"obner bases 
of \textcolor{black}{the} ideals $\langle \bm{F}\rangle$ and $\langle \bm{F}^h\rangle$. 
Let $G$, $G_{\rm hom}$, {and} $G_{\rm top}$ be the reduced Gr\"obner bases of $\langle \bm{F}\rangle$, 
$\langle \bm{F}^h\rangle$, and $\langle \bm{F}^{\rm top}\rangle$, {respectively}, where 
their monomial \textcolor{black}{orderings} are DRL $\prec$ or its extension $\prec^h$. 

As to the computation of $G$, 
in special settings on $\bm{F}$ such as $\bm{F}$ containing field equations 
or $\bm{F}$ appearing in a multivariate polynomial cryptosystem, 
methods using the value $D$ or 
those of the Hilbert function for degrees less than $D$ 
were proposed. (See \cite{ST,Sakata}.) 
Our results in the section can be considered as a {\em certain extension}
and to {\em give exact mathematical proofs} {for the correctness of the methods}. 

{Here, we extend the notion of {\em top part} 
to a homogeneous polynomial $h$ in $R'=R[y]$. 
We call $h|_{y=0}$ the {\em top part} of $h$ and 
denote it by $h^{\rm top}$. 
Thus, if $h^{\rm top}$ is not zero, 
it coincides with the top part $(h|_{y=1})^{\rm top}$ 
of the {dehomogenization} $h|_{y=1}$ of $h$. 
We remark that $g^{\rm top}=(g^h)^{\rm top}$ 
for a polynomial $g$ in $R$.} 

%
\subsection{Gr\"obner basis elements of degree less than $D$}
%
Here we show relations between $(G_{\rm hom})_{<D}$ 
and $(G_{\rm top})_{<D}$ with \textcolor{black}{proofs, which} are useful for 
the computations of $G_{\rm hom}$ and $G$.

Since $\bm{F}^{\rm top}$ is cryptographic semi-regular and $\bm{F}^h$ is $D$-regular by Corollary \ref{cor:Dreg}, 
$H_1(K_{\bullet}(\bm{F}^{\rm top}))_{{<D}}=H_1(K_{\bullet}(\bm{F}^h))_{{<D}}=0$. 
As $H_1(K_{\bullet}(\bm{F}^h))={\rm syz}(\bm{F}^h)/{\rm tsyz}(\bm{F}^h)$ and 
$H_1(K_{\bullet}(\bm{F}^{\rm top}))={\rm syz}(\bm{F}^{\rm top})/{\rm tsyz}(\bm{F}^h)$ 
({see \eqref{eq:H1}}), 
we have the following corollary, where ${\rm tsyz}(\bm{F}^h)$ denotes the module of trivial syzygies (see Definition \ref{def:trivial}).

\begin{corollary}[{\cite[Theorem 1]{Diem2}}]
It follows that ${\rm syz}(\bm{F}^{\rm top})_{<D}={\rm tsyz}(\bm{F}^{\rm top})_{<D}$ and ${\rm syz}(\bm{F}^h)_{<D}={\rm tsyz}(\bm{F}^h)_{<D}$. 
\end{corollary}
This implies that, 
in the Gr\"obner basis computation $G_{\rm hom}$ 
with respect to a \textcolor{black}{graded} ordering $\prec^h$, 
if an S-polynomial $S(g_1,g_2)=t_1g_1-t_2g_2$ 
of degree less than $D$ is reduced to $0$, 
it comes from some trivial \textcolor{black}{syzygy}, that is, 
{$\sum_{i=1}^m(t_1a^{(1)}_i-t_2a^{(2)}_i-b_i)\mathbf{e}_i$} 
belongs to ${\rm tsyz}(\bm{F}^h)_{<D}$, where 
$g_1=\sum_{i=1}^m a^{(1)}_i f_i^h$, 
$g_2=\sum_{i=1}^m a^{(2)}_i f_i^h$, 
\textcolor{black}{and $S(g_1,g_2)=\sum_{i=1}^m b_if_i^h$ is obtained by $\Sigma$-reduction 
in the $F_5$ algorithm (or its variant such as the matrix $F_5$ algorithm) with the Schreyer 
ordering.} 
Thus, since the $F_5$ algorithm (or its variant such as the matrix{-}$F_5$ 
algorithm) with the {\em Schreyer ordering}
automatically discards an S-polynomial whose signature is 
the LM of some trivial syzygy, 
we can avoid unnecessary S-polynomials. 
See {Subsection} \ref{sec:sba} for a {brief} outline of the $F_5$ algorithm 
and the syzygy criterion {(Proposition \ref{prop:syz})}. 

In addition to the \textcolor{black}{above facts}, as mentioned (somehow implicitly) in \cite[Section 3.5]{Bar} and \cite{BFSY}, 
when we compute a Gr\"obner basis of $\langle \bm{F}^h\rangle$ 
for the degree less than $D$ by the $F_5$ algorithm with respect to $\prec^h$, 
for each computed {non-zero 
polynomial $g$ from an S-polynomial, say $S(g_1,g_2)$, 
of degree less than $D$, 
its signature does not come from any trivial syzygy and so 
the reductions of $S(g_1,g_2)$ are done only at its top part.}
{
This implies that the Gr\"obner basis computation process of $\langle \bm{F}^h\rangle$ 
corresponds exactly to that of $\langle \bm{F}\rangle$ for each degree less than $D$, \textcolor{black}{see \cite{KY} for details.}
{Especially, the following lemma holds}. 
Here we give a {\it concrete and easy} proof using Corollary \ref{cor:LM}. 
We {also} note that the \textcolor{black}{argument and the proof of Lemma \ref{lem:LM}} can be considered as \textcolor{black}{corrected versions for those of} \cite[Theorem 4]{Sakata}.
\begin{lemma}\label{lem:LM}
For each degree $d<D$, \textcolor{black}{we have}
\begin{equation}\label{eq:LMD}
{\rm LM}(G_{\rm hom})_{d}
={\rm LM}(G_{\rm top})_{d}.
\end{equation}
\end{lemma}
\begin{proof} We can prove the {equality} \eqref{eq:LMD} by the induction on $d$.
Assume that the {equality} \eqref{eq:LMD} holds for {$d<D-1$.}

Consider any $t\in {\rm LM}(G_{\rm hom})_{d+1}$. Then, there is a polynomial $g\in 
G_{\rm hom}$ such that ${\rm LM}(g)=t$. 
By Corollary \ref{cor:LM}, for $d+1<D$, we have 
\[
(\langle \mathrm{LM}(\langle \bm{F}^h\rangle) \rangle_{R'})_{d+1} 
= (\langle \mathrm{LM}(\langle \bm{F}^{\rm top}\rangle_{R}) \rangle_{R'})_{d+1}
\]
and ${\rm LM}(g)$ is divided by ${\rm LM}(g')$ for {some} $g' \in G_{\rm top}$. 
Since $G_{\rm hom}$ is reduced, ${\rm LM}(g)$ is not divisible by any monomial in ${\rm LM}(G_{\rm hom})_{\leq d}={\rm LM}(G_{\rm top})_{\leq d}$, so that $\mathrm{deg}(g') = d+1$.
Then we have $\mathrm{LM}(g) = \mathrm{LM}(g')$, and so ${\rm LM}(G_{\rm hom})_{d+1}\subset {\rm LM}(G_{\rm top})_{d+1}$. 

By the same argument, ${\rm LM}(G_{\rm hom})_{d+1}\supset {\rm LM}(G_{\rm top})_{d+1}$ 
can be shown. 
We note that for each $t\in {\rm LM}(G_{\rm top})_{d+1}$, there 
is a polynomial $g\in  {(G_{\rm top})_{d+1}}{\subset \langle \bm{F}^{\rm top} \rangle_{d+1}}$ such that $t={\rm LM}(g)$. 
In this case, there are homogeneous polynomials $a_1,\ldots,a_m$ such that $g=\sum_{i=1}^m a_i{\color{black}f_i^{\rm top}}$.
Then $g'=\sum_{i=1}^m a_i{\color{black}f_i^h}$ in $\langle \bm{F}^h\rangle_{d+1}$ has $t$ as its LM. 
\end{proof}

{ 
Next we consider $(G_{\rm hom})_D$. 
\begin{lemma}\label{lem:degD}
For each monomial \textcolor{black}{$t$} in $X$ of degree $D$, 
there is an element $g$ in $(G_{\rm hom})_{\leq D}$ 
such that ${\rm LM}(g)$ divides $t$. Therefore, 
\begin{equation}\label{eq:RD}
\langle {\rm LM}((G_{\rm hom})_{\leq D})\rangle_{\color{black}{R'}} 
\cap R_D=R_D.
\end{equation}
Moreover, for each element $g$ in $(G_{\rm hom})_D$ with 
$g^{\rm top}\not=0$, 
the top-part $g^{\rm top}$ consists of one term, that is, 
$g^{\rm top}={\rm LT}(g)$, where ${\rm LT}$ denotes the leading term 
of $g$. (We recall ${\rm LT}(g)={\rm LC}(g){\rm LM}(g)$.)
\end{lemma}
\begin{proof}
Since $\langle \bm{F}^{\rm top}\rangle_D=R_D$, 
for each monomial $t$ in $X$ of degree $D$, 
there are homogeneous $a_1,\ldots,a_m \in R$ with 
$t=\sum_{i=1}^m a_if_i^{\rm top}$. 
Now consider $h=\sum_{i=1}^m a_if_i^h$\textcolor{black}{, which} belongs to $\langle \bm{F}^h\rangle$. 
Then, as $f_i^h=f_i^{\rm top}+yh_i$ for some $h_i$ in $R'$, 
we have \[
h=\sum_{i=1}^m a_i(f_i^{\rm top}+yh_i)=
\sum_{i=1}^m a_if_i^{\rm top}+y\sum_{i=1}^m a_ih_i
=t+y\sum_{i=1}^m a_ih_i
\]
and ${\rm LM}(h)=t$. As $G_{\rm hom}$ is the reduced Gr\"obner basis 
of $\langle \bm{F}^h\rangle$, 
there is some $g$ in $(G_{\rm hom})_{\leq D}$ whose LM divides \textcolor{black}{${\rm LM}(h)$, as desired}. 

Next we prove the second assertion.
Let $g_1,\ldots,g_k$ be all elements of $(G_{\rm hom})_D$ 
which have non-zero top parts, and set ${\rm LM}(g_1)\prec \cdots 
\prec {\rm LM}(g_k)$. We show that $g_i^{\rm top}={\rm LT}(g_i)$ for 
all $i$.
Suppose, to the contrary, 
that our claim does not hold for some $g_i$. 
Then, $g_i^{\rm top}$ can be written as 
$={\rm LT}(g_i)+ T_2+\cdots +T_s$ for 
some terms $T_2,\ldots,T_s$ in $R_D$. 
Since ${\rm LM}(T_j)\prec {\rm LM}(g_i)$ for $2\leq j\leq s$, it follows from equality \eqref{eq:RD} that each $\mathrm{LM}(T_j)$ is equal to ${\rm LM}(g_\ell)$ for some $\ell<i$ or is divisible by ${\rm LM}(g')$ for some
$g'\in (G_{\rm hom})_{<D}$. This contradicts to the fact that 
$G_{\rm hom}$ is reduced. 
\end{proof}}

{
\begin{remark}
If we apply a signature-based algorithm such as the $F_5$ algorithm or its variant 
to compute the Gr\"obner basis of $\langle \bm{F}^h\rangle$, 
its $\Sigma$-Gr\"obner basis is a Gr\"obner basis, but is 
not always {\em reduced} in the sense of ordinary Gr\"obner 
basis, in general.  In this case, we have to compute so called {\em inter-reduction} among 
elements of the $\Sigma$-Gr\"obner basis for obtaining the reduced Gr\"obner basis. 
\end{remark}
}

\subsection{Gr\"obner basis elements of degree not less than $D$}\label{subsec:overD}

{In this subsection, we shall
extend the upper bound on solving degree given in \cite[Theorem 2.1]{ST} 
to our case.}

\begin{remark}
In \cite{ST}, polynomial ideals over ${\mathbb F}_q[X]$ are considered. 
Under the condition where the generating \textcolor{black}{sequence {$\bm{F}$}} contains the field equations $x_i^q-x_i$ for $1 \leq i \leq n$, 
recall from Theorem \ref{thm:tenti} {(\cite[Theorem 6.5 \& Corollary 3.67]{Tenti})} that {the solving degree $\mathrm{sd}_{\prec}^{\rm hsd}(\bm{F})$} {with respect to a Buchberger-like algorithm for} $\langle \bm{F}\rangle$ is {upper-}bounded by $2D-2$, where $D=d_{\rm deg} (\langle \bm{F}^{\rm top} \rangle)$. 
{In the proofs of \cite[Theorem 6.5 \& Corollary 3.67]{Tenti}, the property 
$\langle \bm{F}^{\rm top}\rangle_{D}=R_D$ was essentially used for obtaining the upper-bound.}
As the property also holds in our case, we {may} apply their arguments. 
Also in \cite[Section 3.2]{BNDGMT21}, the case where $\bm{F}^h$ is cryptographic semi-regular is considered. The results on the solving degree and the maximal degree 
of the Gr\"obner basis are heavily related to our result in this subsection. 
\end{remark}

{
\textcolor{black}{Now we show an upper-bound on the solving degree of $\bm{F}$ by 
using the set $H:=\{g|_{y=1}: g\in (G_{\rm hom})_{\leq D}\}$, that is, 
at the pre-process of the computation of $G$, we first compute $H=(G_{\rm hom})_{\geq D}$, 
and at the latter process, we continue the computation from $H$.}
We remark that, when we use the normal selection strategy 
on the choice of S-polynomials, 
the Gr\"obner basis computation of $\langle \bm{F}\rangle$ 
proceeds along with the graded structure of $R$ in its early \textcolor{black}{stages}, 
By Lemma \ref{lem:LM} 
it simulates faithfully that of $\langle \bm{F}^{\rm top}\rangle$ 
until the degree of computed polynomials becomes $D-1$, 
that is, it {produces} 
$\{g|_{y=1} : g\in (G_{\rm hom})_{<D}\}$. 
Also, by Lemma \ref{lem:degD}, 
it \textcolor{black}{may} also {produce} 
$\{g|_{y=1} : g\in (G_{\rm hom})_{D},\ { g^{\rm top}\not=0}\}$ 
\textcolor{black}{by carefully choosing S-polynomials, see \cite{KY} for details.}
We {also} 
note that {the} $F_5$ algorithm actually uses the normal strategy.}

\begin{lemma}\label{lem:sd}
{If $D\geq {\rm max}\{\deg(f) : f\in \bm{F}\}$, then the} 
maximal Gr\"obner basis degree and the solving degree \textcolor{black}{${\rm sd}^{\rm hsd}_\prec (\bm{F})$ (see Subsection \ref{subsec:complexity} for the definition of ${\rm sd}^{\rm hsd}_\prec (\bm{F})$)} are bounded as follows:
\[
{\rm max.GB.}\deg_\prec (\bm{F}) \leq D \mbox{ and }
\color{black}{{\rm sd}^{\rm hsd}_\prec (\bm{F})} \leq {2D-2}.
\]
\end{lemma}
\begin{proof}
Recall from Lemma \ref{lem:degD} that $\langle {\rm LM}(H)\rangle$ 
contains all monomials in $X$ of degree $D$. 
We continue the Gr\"obner basis computation from $H$. 
In this {\em latter process}, all polynomials generated from S-polynomials 
are reduced by elements of $H$. Therefore, 
their LM's are reduced with respect to any monomial (in $X$) of degree $D$ and thus, their degrees are not more than $D-1$. 
Thus, the maximal Gr\"obner basis degree is \textcolor{black}{upper-}bounded by $D$, 
and the degree of S-polynomials dealt in the whole computation 
is \textcolor{black}{upper-}bounded by $2D$. 

{
Next we show that we can avoid any S-polynomial of degree $2D$ 
or $2D-1$.} 

\noindent
(i) If an S-polynomial 
$S(g_1,g_2)$ has its degree $2D$, 
then $\deg(g_1)=\deg(g_2)=D$ and $\gcd({\rm LM}(g_1),{\rm LM}(g_2))=1$.
Then, Buchberger's criterion 
predicts that $S(g_1,g_2)$ is always reduced to 0. 

\noindent
(ii) If an S-polynomial 
$S(g_1,g_2)$ has its degree $2D-1$, 
then $\deg(g_1)=\deg(g_2)=D$, $\deg(g_1)=D, \deg(g_2)=D-1$ 
or $\deg(g_1)=D-1, \deg(g_2)=D$. 

For the case where $\deg(g_1)=D, \deg(g_2)=D-1$ 
or $\deg(g_1)=D-1, \deg(g_2)=D$, we have
$\gcd({\rm LM}(g_1),{\rm LM}(g_2))=1$, and hence
$S(g_1,g_2)$ is always reduced to 0 by Buchberger's criterion.

Finally, we consider the remaining case where $\deg(g_1)=\deg(g_2)=D$. 
In this case, $g_1$ and $g_2$ should belong to $H$ and recall from Lemma \ref{lem:degD} that both of $(g_1)^{\rm top}$ and $(g_2)^{\rm top}$ are single terms. 
Then $S(g_1,g_2)$ cancels the top parts of $t_1g_1$ and $t_2g_2$, 
where $S(g_1,g_2)=t_1g_1-t_2g_2$ for some terms $t_1$ and $t_2$. 
Thus, the degree of $S(g_1,g_2)$ is less than $2D-1$. 
\end{proof}

{\begin{remark}
We refer to \cite[Remark 15]{CG20} for another proof of ${\rm max.GB.}\deg_\prec (\bm{F}) \leq D$.
We also note that, 
if $D=d_{\rm reg}(\bm{F}^{\rm top})$ 
is finite, Lemma 3 and Lemma 4 
hold without the assumption that 
$\bm{F^{\rm top}}$ is cryptographic 
semi-regular. 
\end{remark}}

As to the computation of $G_{\rm hom}$, 
we have a result similar to Lemma \ref{lem:sd}.
Since $\langle {\rm LM}(G_{\rm hom})_{\leq D}\rangle$ contains 
all monomials in $X$ of degree $D$, 
for any polynomial $g$ generated in the middle of the computation of $G_{\rm hom}$
the degree of \textcolor{black}{the $X$-part} of ${\rm LM}(g)$ is less than $D$. 
Because $g$ is reduced by $(G_{\rm hom})_{\leq D}$. 
Thus, letting ${\mathcal U}$ be the set of 
all polynomials generated during the computation of $G_{\rm hom}$, 
we have 
\[
\{\mbox{The $X$-part of } {\rm LM}(g): g\in {\mathcal U}\}
\subset \{x_1^{e_1}\cdots x_n^{e_n} : e_1+\cdots+e_n\leq D\}. 
\]
As different $g,g'\in {\mathcal U}$ can not have the same $X$ part in their leading terms, 
the size $\# {\mathcal U}$ is upper-bounded by the number of monomials in $X$ of degree not greater 
than $D$, that is $\binom{n+D}{n}$.
By using the $F_5$ algorithm or its efficient variant, 
under an assumption that every unnecessary S-polynomial can be avoided, 
the number of computed S-polynomials during the computation 
of $G_{\rm hom}$ coincides with the number $\#{\mathcal U}$ 
and is \textcolor{black}{upper-}bounded by $\binom{n+D}{n}$. 

\begin{example}\label{ex:sd}
When $m=n+1$ and $d_1=\cdots =d_m=2$, the Hilbert-Poinc\'are series 
of $R/\langle \bm{F}^{\rm top}\rangle$ is $\left[\frac{(1-z^2)^{n+1}}{(1-z)^n} \right]$. 
Since $\frac{(1-z^2)^{n}}{(1-z)^n}=(1+z)^n=\sum_{i=0}^n \binom{n}{i} z^i$, 
we have 
\[\frac{(1-z^2)^{n+1}}{(1-z)^n}=(1+z)^n(1-z^2)
=
1+nz+\sum_{i=2}^n \left(\binom{n}{i} - \binom{n}{i-2}\right) z^i
-nz^{n+1}-z^{n+2},
\]
so that $D=d_{\rm reg}(\langle \bm{F}^{\rm top}\rangle) = \min \left\{ i : \binom{n}{i} -{\binom{n}{i-2}}\leq 0 \right\} = \lfloor (n+1)/2\rfloor+1$, \textcolor{black}{see \cite[Theorem 4.1]{BNDGMT21}}.
In this case, it follows from
\[
2D-2= 2(\lfloor (n+1{\color{black})}/2\rfloor+1)-2 =
\begin{cases}
n+1 & \mbox{($n$: odd),}\\
n &  \mbox{($n$: even)}
\end{cases}
\]
that ${\rm sd}^{\rm hsd}_\prec (\bm{F})\leq {n+1}$ in Lemma \ref{lem:sd};
see \cite[{Theorem 4.2, Therorem 4.7}]{BNDGMT21} for the bound in the case where $\bm{F}^h$ is a generic sequence.

We note that, in the homogeneous case, the solving degree ${\rm sd}^{\rm hsd}_{\prec^h}(\bm{F}^h)$ is equal to the maximal Gr\"obner basis degree of $\bm{F}^h$ (for an appropriate setting in the algorithm one adopts), so that we can apply Lazard's bound, see Theorem \ref{thm:Lazard_zero_dim}.
It also follows (see \cite{KY} for details) that the solving degree ${\rm sd}_{\prec}^{\rm hsd}(\bm{F})$ can be \textcolor{black}{upper-}bounded by ${\rm max.GB}.\deg_{\prec^h}(\bm{F}^h)= {\rm sd}^{\rm hsd}_{\prec^h}(\bm{F}^h)$, and we can apply Theorem \ref{thm:Lazard_zero_dim}, as our case satisfies its conditions.
Then, for the case where $m=n+1$ and $d_1=\cdots =d_{n+1}=2$, 
Lazard's bound gives the bound $n+2$ for ${\rm max.GB}.\deg_{\prec^h}(\bm{F}^h)=
{\rm sd}^{\rm hsd}_{\prec^h}(\bm{F}^h)\geq {\rm sd}_{\prec}^{\rm hsd}(\bm{F})$.
\end{example}

\begin{acknowledgement}{
The authors thank the anonymous referee for helpful comments.
The authors are also grateful to Yuta Kambe and Shuhei Nakamura for helpful comments.
This work was supported by JSPS Grant-in-Aid for Young Scientists 20K14301 and 23K12949, JSPS Grant-in-Aid for Scientific Research (C) 21K03377, and JST CREST Grant Number JPMJCR2113.
}
\end{acknowledgement}

\renewcommand{\baselinestretch}{0.85}

\end{document}